\documentclass[a4paper,UKenglish,cleveref, autoref, thm-restate]{lipics-v2021}

\usepackage{algorithm}
\usepackage[noend]{algpseudocode}
\usepackage{pifont}
\usepackage{bbm}
\usepackage{xfrac} 

\definecolor{ao}{rgb}{0.0, 0.5, 0.0}
\newcommand{\greencheck}{{\color{ao}\ding{52}}}
\newcommand{\redcross}{{\color{red}\ding{56}}}

\usepackage{xspace}
\newcommand{\name}{\textsc{BigDipper}\xspace}
\newcommand{\nameOne}{{\textsc{BigDipper-}$7$}\xspace}
\newcommand{\nameLite}{\textsc{BigDipper-Lite}\xspace}
\newcommand{\nameQuarter}{\textsc{BigDipper-}$\sfrac{1}{4}$\xspace}
\newcommand{\teaspoon}{\textsc{Teaspoon}\xspace}

\newcommand{\navigation}{\textsc{Navigator}\xspace}
\newcommand{\scooping}{\textsc{Scooping}\xspace}
\newcommand{\DAL}{\textmd{DA-CR}\xspace}

\newcommand{\DA}{\textmd{DA-CR}\xspace}
\newcommand{\card}{\textsc{Card}\xspace}
\newcommand{\cardOne}{{\textsc{Card-}$7$}\xspace}
\newcommand{\cardQuarter}{\textsc{Card-}\sfrac{1}{4}\xspace}

\newcommand{\cardLite}{\textsc{Card-Lite}\xspace}

\newcommand{\LSC}{\textsc{LSC}\xspace}

\newcommand{\cl}{\textit{commitment list}\xspace}
\newcommand{\al}{\textit{attestation set}\xspace}
\newcommand{\mycomment}[1]{}
\newcommand{\primitive}[1]{\textsf{{#1}}}
\newcommand{\event}[1]{\textsf{{#1}}}
\newcommand{\proceduref}[1]{\textbf{\textsc{{#1}}}}
\newcommand{\property}[1]{\textup{{#1}}}
\newcommand{\primitiveTl}[1]{\textit{{#1}}}
\newcommand{\messageProto}[1]{\textrm{{#1}}}

\newcommand{\kzgcommitment}{C}

\algdef{SE}[EVENT]{Event}{EndEvent}[1]{\textbf{upon}\ #1\ \algorithmicdo}{\algorithmicend\ \textbf{event}}%
\algtext*{EndEvent}
\newcommand{\false}{\texttt{false}\xspace}
\newcommand{\true}{\texttt{true}\xspace}
\nolinenumbers

\usepackage{hyperref} 

\title{ $\name$: Sharded Censorship Resistant Data Availability for Leader-Based BFT} 

\author{Bowen Xue}{EigenLabs,  Seattle, WA, 98115, USA}{}{[orcid]}{}

\author{Samuel Laferriere}{Independent,  NY, USA}{
samlaf92@gmail.com}{[orcid]}{} 

\author{Souhbik Deb}{EigenLabs,  Seattle, WA, 98115, USA}{}{[orcid]}{} 

\author{Sreeram Kannan}{University of Washington,  Seattle, WA, 98115, USA}{ksreeram@uw.edu}{[orcid]}{}

\authorrunning{B. Xue, S. Laferriere, S. Deb, and S. Kannan} 

\Copyright{Bowen Xue, Samuel, Laferriere, Souhbik Deb, and Sreeram Kannan} 

\ccsdesc[500]{Security and privacy~Distributed systems security}

\keywords{Censorship Resistance; byzantine fault tolerant; consensus; blockchain}

\category{} 

\relatedversion{} 

\ArticleNo{48}

\begin{document}

\maketitle

\begin{abstract}
Leader-based Byzantine-fault-tolerant (BFT) protocols provide low latency and
simple communication structure, but they give the leader short-term control over
transaction inclusion. A malicious leader can keep the protocol live while
delaying or excluding time-sensitive transactions such as auction bids, oracle
updates, liquidations, and bridge messages. Existing responses often build a
fixed censorship-resistance, hiding, or ordering mechanism into the protocol
path, forcing all transactions to pay for the same protection level.
\name follows the end-to-end principle: the consensus layer exposes inclusion
primitives rather than hardcoding stronger policies. Higher-layer protocols can
then choose their own submission strategies and resources, whether through
replication, erasure coding, or other mechanisms, to obtain the
censorship-resistance, hiding, ordering, or execution guarantees they need.
At the core of \name is censorship-resistant data availability, or DA-CR, which
certifies available replica-contributed mini-blocks for use by leader-based
consensus. A central design goal is that data remains sharded on the consensus
critical path: validators do not reconstruct or execute the full payload before
voting, but instead check commitments, availability evidence, and the DA-CR
inclusion rule.
We define DA-CR guarantees for data-tampering resistance, honest mini-block
inclusion, and residual leader influence. We then give concrete constructions
based on erasure coding and linear commitments, analyze client-tunable
transaction submission, and instantiate \name inside HotStuff-2.
\end{abstract}

\section{Introduction}
\label{sec:introduction}

Leader-based Byzantine-fault-tolerant (BFT) protocols are the dominant design
point for many production blockchain systems because they provide low latency,
simple communication structure, and high throughput in partially synchronous
networks~\cite{yin2019hotstuff,buchman2018latest,cryptoeprint:2023/397,castro1999practical}.
In these protocols, the leader coordinates the consensus round and constructs
the proposal that replicas vote on. This structure is efficient, but it gives
the leader short-term control over inclusion: a malicious leader can keep the
protocol live while delaying or excluding selected transactions whose value
depends on timely inclusion. This problem is not captured by ordinary BFT
liveness, which only guarantees that the protocol eventually commits blocks,
not that transactions received by honest replicas are included promptly. The
same concern has been studied in recent work~\cite{mcp,cr_thput}. For many
blockchain applications, this gap is economically meaningful: auction bids,
oracle updates, liquidations, bridge messages, and other time-sensitive
transactions may lose most of their value if delayed by only a few blocks. We
refer to this problem as \emph{short-term censorship}: the leader need not halt
the protocol indefinitely to cause harm; it only needs to prevent prompt
inclusion during the relevant window.

A natural response is to add stronger anti-censorship machinery to the core
protocol path. Fair-ordering protocols constrain transaction order using
replica receive times~\cite{kelkar2021themis,kelkar2022order}. MCP-style
multi-proposer protocols use erasure-coded DA and relay attestations
to prevent a leader from selectively including available batches~\cite{mcp}.
These approaches provide valuable guarantees for applications that need them.
However, they also choose a strong protection level for the protocol path:
fair-ordering systems impose an ordering policy, while MCP-style systems attach
nonselective inclusion and hiding machinery to the slot pipeline. Not every
transaction requires this level of short-term censorship resistance or
pre-confirmation hiding.

This paper argues for a more modular design. Not every transaction requires the
same level of short-term censorship resistance. A liquidation, final auction
bid, or bridge message may justify higher dissemination cost. An ordinary
transfer, delayed message, or low-value application action may not. A consensus protocol should therefore expose
a clean mechanism for short-term censorship-resistant inclusion, while allowing
clients and applications to decide how much protection is worth its cost.

This view follows the end-to-end argument:
functions whose correct implementation depends on application-level
requirements should not be unnecessarily embedded in a lower-level
subsystem~\cite{saltzer1984end}. Applied to BFT consensus, the middle layer
should provide the minimal common mechanism needed for safe ordering, available
data, and enforceable inclusion. Application specific policies such as how
aggressively to resist short-term censorship, whether to hide transaction
contents before confirmation, which ordering rule to use, and how to execute or
serve full data, should remain modular choices at the endpoints or upper
layers.

BigDipper applies this principle through a censorship-resistant data
availability layer, denoted $\DAL$, placed between transaction submission and
leader-based BFT consensus. Replicas collect submitted transactions into
mini-blocks, and $\DAL$ establishes which mini-block data is available and must
be captured by a valid leader proposal. The abstraction captures the
concerns about
data tampering, guaranteed inclusion of honest replica contributions, and the
leader's residual influence over which available data is captured. We later
formalize these dimensions using parameters $(t,\eta)$.

With this abstraction in place, probabilistic short-term censorship resistance becomes a tunable opt-in property rather than a fixed protocol tax. A client that cannot tolerate
short-term censorship submits its transaction to more replicas. A client that
is less sensitive to delay can submit to fewer replicas and pay lower
communication cost. The protocol does not force one protection level on all
transactions: clients choose the prompt-inclusion probability they need.

An additional property emerges from the interaction between transaction
submission and $\DAL$. A malicious leader may still censor some transactions,
especially those submitted with small dissemination fanout. However, censoring
a large fraction of transactions becomes much harder: as the total number of
submitted transfactions grows, the probability that an adversary can censor an
$\alpha$ fraction of them decreases exponentially. Our analysis formalizes this
effect with a union-bound argument. Thus BigDipper does not make every
transaction uncensorable at no cost; instead, it exposes an explicit
submission-level cost/security tradeoff while making broad short-term
censorship increasingly unlikely at scale.
The detailed $\DAL$ construction, transaction-submission protocol, and BFT
integration are introduced later.

BigDipper applies the same modular principle to the boundary between consensus
and execution. In many blockchain systems, consensus validators are also
expected to receive, store, and often execute the entire block payload. This
coupling makes the system easy to reason about, but it also forces every
validator to pay the cost of full-data execution. 
For consensus, a validator only needs to check that the proposed block is backed
by available, certified data and that the leader has satisfied the $\DAL$
inclusion rule.

BigDipper separates these responsibilities. The consensus layer orders and
certifies available data, while the upper layer decides how that data is
retrieved, derived, executed, and served to users. This is similar in spirit to
rollup-style architectures, where data publication and execution need not be
performed by the same entities. If an application wants conventional full
replication, it can still require every node to execute the full payload. But
BigDipper does not force that choice at the consensus layer.

This separation opens a broader execution design space. An application may use
high-capacity full-data executors, a staked executor committee, trusted
execution environments, zk-proved execution, or optimistic execution with fraud
proofs. These choices have different trust, latency, and cost tradeoffs. The
point is that they are execution-layer decisions. The consensus validator's
critical path remains sharded: validators check commitments, availability
evidence, and the $\DAL$ inclusion rule, rather than reconstructing and
executing the full payload before voting.

BigDipper makes this possible by combining erasure coding with a commitment
scheme that preserves linear structure. Replicas contribute mini-blocks, the
leader erasure-codes them into a sharded data object, and validators check
compact commitment evidence for the shards they receive. Because both the
encoding and the commitment relation are linear, these local checks still bind
the leader to a single global data object.
As a result, no individual validator needs to download the entire data object
merely to check the short-term censorship-resistance condition.

The sharded design also requires care at the transaction-submission layer. If
clients were encouraged to broadcast every transaction to every validator, then
the system would recreate full replication at the edge and lose much of the
benefit of sharded consensus. BigDipper instead treats fanout as part of the
client's cost/security choice. Clients choose how many replicas to contact
based on how much short-term censorship risk they are willing to tolerate. A
transaction that needs stronger prompt-inclusion protection can pay the
data-availability cost of reaching more replicas; a less time-sensitive
transaction can use smaller fanout and pay less. In this way, the cost of
stronger short-term censorship resistance is borne by the transactions that
need it, rather than imposed uniformly on all traffic.
The transaction-submission interface is intentionally simple: clients choose
how much effort to spend to reach replicas. More sophisticated coded
submission protocols, such as Sedna\cite{sedna} transaction sharding, can be viewed
as richer strategies at the same boundary, by using erasure code to reduce redundancy while boosting censorship resistance.

This gives BigDipper a different scalability profile from full replication. In
a fully replicated design, adding validators improves fault tolerance but does
not increase data capacity, because every validator must still process the same
full payload. In BigDipper, the DA-CR layer can capture mini-blocks contributed
by many replicas, so the amount of data made available to consensus can grow
with aggregate replica-side bandwidth. This gives the consensus layer a
sharded, linear data-contribution path. However, this is not by itself an
end-to-end throughput claim: overall throughput and latency also depend on how
the upper-layer execution and serving system is designed. A fully replicated
executor layer may become the limiting factor, while specialized executors,
committees, TEEs, zk-proved execution, or optimistic execution expose different
performance and trust tradeoffs.

Finally, BigDipper gives a concrete integration with a modern leader-based BFT
protocol. We instantiate the ordering layer using HotStuff-2~\cite{cryptoeprint:2023/397}.
The integration is surgical: BigDipper does not replace the safety core of
HotStuff-2, nor does it require all validators to fully replicate block data. A
leader proposal is valid only if it incorporates the required $\DAL$ output,
and honest validators check this condition before voting. This shows that
$\DAL$ is not merely a high-level modular abstraction; it can be connected
precisely to an existing leader-based BFT protocol.

\subsection{Contributions}

This paper makes the following contributions.

\begin{enumerate}
    \item \textbf{A modular abstraction for short-term censorship-resistant inclusion.}
    We introduce $\DAL$, a censorship-resistant data availability abstraction
    between transaction dissemination and leader-based BFT consensus. $\DAL$
    captures the point at which replica-contributed mini-blocks become
    available and must be included by the consensus leader. We characterize
    $\DAL$ protocols using parameters $(t,\eta)$, which capture
    data-tampering resistance, guaranteed honest mini-block inclusion, and the
    leader's residual influence over inclusion.

    \item \textbf{A sharded, scalable data path for leader-based BFT.}
    We present a design in which ordinary validators do not need to download the
    entire block payload merely to vote. Validators vote using commitments,
    metadata, and $\DAL$ evidence, allowing aggregate proposal bandwidth to grow
    with replica-side data contribution rather than being limited by a single
    leader's input bandwidth. This design separates consensus from full-data
    execution and serving, allowing specialized upper-layer nodes to download,
    verify, execute, and serve data without requiring every consensus validator
    to do so.

    \item \textbf{Client-tunable short-term censorship resistance.}
    We design \textsc{Navigator}, a transaction dissemination mechanism that
    lets clients amplify prompt-inclusion probability by sending a transaction
    to multiple replicas. Clients that send to more replicas obtain stronger
    short-term inclusion protection; clients that send to fewer replicas pay
    lower communication cost but accept a higher probability of short-term
    censorship. We analyze this cost/security tradeoff, including a union-bound
    argument showing that broad censorship becomes unlikely when sufficiently
    many clients send to sufficiently many replicas.

    \item \textbf{A concrete HotStuff-2 integration and analysis.}
    We instantiate $\name$ inside HotStuff-2 through the \scooping rule. The
    integration specifies the modified protocol flow and proves that safety and
    liveness are preserved while enforcing inclusion of the $\DAL$ output. We
    also analyze how $\name$ reduces the single-leader proposal bottleneck by
    allowing multiple replicas to contribute mini-blocks concurrently.
\end{enumerate}

\subsection{Paper Organization}

Section~\ref{sec:background} gives background on BFT protocols and data
availability. Section~\ref{sec:card_property} defines the $\DAL$ abstraction and
its properties. Section~\ref{sec:card_design} presents concrete $\DAL$
constructions. Section~\ref{sec:navigation} gives the \textsc{Navigator}
transaction broadcast protocol. Section~\ref{sec:bigDipper_hotstuff} describes the
\scooping rule and the HotStuff-2 integration. The appendices contain additional
protocol variants, proofs, and discussion of ordering extensions.   

\subsection{Related Approaches and Design Position}
\label{sec:introduction:related_position}

MCP studies short-term censorship resistance in multi-proposer BFT using
availability and relay attestations, preventing a leader from selectively
including visible transaction batches~\cite{mcp}. It also situates this design
relative to other leaderless, DAG-based, and multi-proposer BFT protocols.
BigDipper takes a different architectural position. Following the end-to-end
principle, it does not bake hiding or one fixed censorship-resistance level into
the consensus path. Instead, the $\DAL$ layer enforces non-hiding inclusion of
replica-contributed data, while the upper layer chooses how much
submission cost, hiding, or ordering policy.

Sedna further supports this separation by showing that the client-side
dissemination layer is itself a rich design space: it uses erasure-coded
transaction fragments to improve the bandwidth, latency,
censorship-resistance, and hiding tradeoff in MCP-style systems~\cite{sedna}.
In BigDipper, {\navigation} is the minimal submission interface: the degree of
short-term censorship resistance depends on how many honest replicas receive
the transaction before the target {\DAL} instance. Sedna-style coded submission
could be composed with BigDipper by sending erasure-coded fragments rather than
full transaction copies, reducing duplicate data while preserving the same
fanout-based inclusion logic. We leave this extension to future work; the key
point is that such dissemination choices are upper-layer strategies built on
top of the DA-CR interface, not policies hardcoded into consensus.

Other BFT protocols also obtain inclusion from many replica contributions, often
through leaderless, all-to-all, or DAG-style dissemination. We do not repeat the
full comparison here; Appendix~\ref{sec:appendix:bft_complexity_analysis}
summarizes representative protocols and their communication tradeoffs.
BigDipper targets a different point: it keeps an active leader-based BFT core
while adding a sharded DA-CR layer that forces the leader to capture
replica-contributed mini-blocks.

\section{Background and Related Works}
\label{sec:background}

\subsection{BFT problem }
\label{sec:background_bft}
A BFT protocol is made of $n$ replicas that maintain a consensus on some common state. Clients submit transactions and rely on honest replicas to process their transactions correctly.
Replicas communicate among themselves on a network, which is modeled as asynchronous, partially synchronous or synchronous. In partially synchronous networks, there is a notion of an unknown Global Stabilization Time (GST), after which all transactions can arrive within a known $\Delta$ duration.
$\name$ is a modular system designed for BFT protocols based on the partially synchronous network where a leader can coordinate replicas.
A BFT protocol must satisfy the following properties:
\begin{itemize}
\item \textbf{Safety}: At any time, the transactions ledger of every pair of honest replicas is the prefix of another.
\item \textbf{Liveness}: After GST, any valid transaction by honest client is eventually replicated by honest replicas.
\end{itemize}
The lower bound with $f$ malicious replicas requires $n\ge3f+1$ ~\cite{dwork1988consensus}.
Every round of a BFT protocol is associated with a view number. Appendix~\ref{sec:appendix:background:validity} discusses the validity property.
%

\subsection{DA and Verifiable Information Dispersal}
\label{sec:background:da-vid}
Data availability (\textbf{DA}) is an implied requirement by all BFT protocols such that all replicas have access to a common ordered data at any time.
A {\DAL} also provides the data availability. 
Most leader based BFT protocol uses a simple DA by having the leader broadcast the data to everyone.
There are more sophisticated designs of DA, including AVID~\cite{cachin2005asynchronous}, AVID-FP~\cite{hendricks2007verifying}, AVID-M~\cite{yang2022dispersedledger}, SEMI-AVID-PR~\cite{nazirkhanova2021information}. We compare our censorship resistant DA (i.e. {\DA}) with them in Appendix~\ref{sec:appendix:card:avid}.
Information dispersal is a technique that uses Reed-Solomon code to split data into $O(n)$ chunks, and each replica only stores a constant number of chunks to ensure the data is reconstructable. The verifiability ensures two clients always retrieve identical data.

\subsection{Reed Solomon Erasure code}
\label{sec:background_rs}
A Reed-Solomon(RS) erasure code is specified by a pair $(k,h)$, where $k$ is the number of systematic chunks, created from the input data, and the second parameter $h$ is the coding redundancy. The newly generated data are called parity chunks.
The ratio between $k$ and $h$ is called the coding ratio.
The encoding process of RS code is based on the idea of polynomial interpolation\cite{reed1960polynomial}, which is a linear operation.
Given an input array, the encoding can be implemented so that the output is a concatenation of input data and parity chunks.
We use $rsEncode, rsDecode$ to refer to the encoding and decoding operations. 

\subsection{Polynomial Commitment and 1D KZG}
\label{sec:background:poly-commit}
A polynomial commitment scheme~\cite{kate2010constant} treats data as a polynomial and provides primitives including \primitiveTl{commit}, \primitiveTl{createWitness}, and \primitiveTl{verifyEval}. 
The \primitiveTl{commit} primitive provides a concise and unique representation of a polynomial.
It has a \textsf{Polynomial binding} property such that it is computationally infeasible to find two polynomials that commit to the same commitment. 
With primitives \primitiveTl{createWitness}, \primitiveTl{verifyEval}, the scheme allows us to reveal(open) evaluations at specific indices along with witnesses(proofs) without disclosing the entire polynomial. When presented with a correct witness, any verifier can use \primitiveTl{verifyEval} and be convinced that the evaluations are indeed correct against the commitment.
The scheme has a \textsf{Evaluation binding} property such that it is computationally infeasible to find two witnesses for two different polynomial evaluations at the same evaluation index that passes the \primitiveTl{verifyEval}. 
Kate-Zaverucha-Goldberg (KZG)\cite{kate2010constant} is a polynomial commitment scheme that provides linearity in commitments and witnesses.
Appendix~\ref{sec:appendix:background:kzg} contains more technical definitions. 
KZG commitment of bn254 curve on 1MiB data takes about 100ms on Apple M4 Pro with 12 working threads\cite{kzgcommit}.

%

\subsection{2D RS KZG} 
\label{sec:background_rs_2dkzg}
Suppose we have $n$ polynomials, each of degree $d-1$, and we want to encode it with a coding ratio of $\frac{1}{3}$.
A matrix can be created where every column contains one polynomial. 
The RS encoding is applied to each of $d$ row by extending the evaluations to $2n$ more points. 
The commitment of every newly generated $2n$ polynomials can be shown to be the polynomial extension of the first $k$ column commitments. It is due to the linearity property of KZG and RS~\cite{Vitalik2d}.
Figure~\ref{fig:card_data_struct} provides a visualization. More discussion can be found in Appendix~\ref{sec:appendix:background:rs}.

\subsection{Combined Signature Scheme and CRHF}
\label{sec:background:signature}
We use combined signature to refer to a signature generated from either a multi-signature or threshold signature scheme. We will use the exact name when a specific signature scheme is used.
On paring friendly curves, a practical way to aggregate multiple signatures on a common message into one signature is to use BLS multi-signature \cite{boneh2018compact}. 
The signature scheme provides three secure primitive $\sigma_i \gets$ \primitiveTl{ms-sign}$(sk_i,m)$; $\sigma, I_{\sigma} \gets$ \primitiveTl{ms-agg} $(pk_1,\sigma_1 \cdots $ $ pk_ ,\sigma_t)$; $bool \gets \primitiveTl{ms-verify}(I_{\sigma}, m, \sigma)$, where $I_{\sigma}$ is an indicator vector for the signers, which has a linear size. 
In contrast to multi-signature, a threshold signature show out of $n$ signers at least some ratio number of signers have signed a common message. Its verification does not require an indicator, $I_{\sigma}$. So the verification requires only constant message. 
To convert a list of data into a single message, we assume a collision resistant hash function(\primitiveTl{CRHF}) with negligible probability for two different data to hash into the same digest. 
A combined signature is valid if it contains $n-f$ replicas signatures. 

\subsection{Security Assumptions for {\name}}

We model Byzantine messages as being delivered as early as the adversary wants,
subject to the network model. In probability calculations involving message
arrival order, we pessimistically assume Byzantine replicas' messages are
available to the leader before honest replicas' messages. This convention gives
the adversary priority in leader collection and is used only to obtain
worst-case inclusion bounds; Byzantine replicas may still delay, omit, or
equivocate messages arbitrarily.

The network is partially synchronous: after GST, messages between honest
replicas are delivered within a known bound $\Delta$\cite{dwork1988consensus}. Up to $f$ replicas are
Byzantine and may be coordinated by a single adversary or act separately.
Byzantine replicas may send, delay, omit, or equivocate messages arbitrarily.
When analyzing leader collection and transaction inclusion probabilities, we
use a pessimistic scheduling convention in which Byzantine messages can arrive
before honest messages. This gives the adversary maximum priority in filling the
leader's collection set and avoids relying on favorable Byzantine latency.

Transaction censorship occurs when Byzantine replicas, including a Byzantine
leader, intentionally exclude transactions. This is distinct from a transaction
missing a mini-block because an honest replica did not receive it before the
collection deadline. Honest clients may submit transactions or retrieve
dispersed data, and no honest client is assumed to have an offline trust
relationship with any particular BFT replica.

\section{Censorship Resistant DA }
\label{sec:card_property}

The goal of this section is to define the properties of the $\DAL$ layer
independently from the later protocol construction and BFT integration. In
$\name$, short-term censorship resistance is obtained by composing three
components: a transaction-submission mechanism that delivers client transactions
to replicas, a $\DAL$ layer that lets replicas contribute mini-blocks and makes
the resulting data available, and a BFT integration that forces the leader to
consume the $\DAL$ output when forming a block. This section focuses on the
second component. The $\DAL$ layer is the point at which decentralized
transaction submission is converted into a certified data object, and its
properties determine how much short-term censorship resistance the surrounding
BFT protocol can obtain. We characterize a $\DAL$ protocol along two
dimensions: guaranteed inclusion of honest replica
mini-blocks, and the leader's ability to capture excluded replica space. These
dimensions are summarized by the parameters $(t,\eta)$.

\subsection{Primitives and Definitions}
\label{sec:card_propery:primitives_definitions}
{\name} is built on three protocol components in order to provide short term censorship resistance for leader based BFT protocols. The first component delivers clients'
transactions to BFT replicas; {\DAL} is the second component, it creates a mechanism that allows all replicas to create mini-blocks, and store them in a distributed manner;
the last component integrates the {\DAL} with a BFT protocol by enforcing the current leader to consume all data from {\DAL} when producing any BFT block.
%
The {\DAL} is a critical piece in the {\name} BFT system as its properties determine how much censorship resistance the system has. In this section, we first state interfaces of two invocations to {\DA}. 
To precisely describe their properties, we provide definitions on three aspects of censorship resistance. In the end we develop a parameterization of any {\DA} protocols with three numbers.


%
The {\DAL} provides a {\primitive{Disperse}} invocation that distributes the data among replicas, and one {\primitive{Retrieve}} invocation to safely retrieve the data from honest replicas.
For censorship resistance, we require that the {\primitive{Disperse}} invocation is initiated in a distributed manner from at least $n-f$ replicas,  
and the final dispersal data $B$ assembled by the leader must also contain at least $n-f$ mini-blocks. 
%
For the {\primitive{Retrieve}} invocation, either replicas or clients can invoke the procedure to retrieve partial or full dispersed data. %

Every invocation needs to be initiated against a {\DAL} instance with a unique ID, which is used by replicas and clients to identify the context for intended invocations.
%
%
%
Every {\primitive{Disperse}} invocation is associated with two events, a $\event{Start}$ event and a $\event{Finish}$ event. 
The $\event{Finish}$ event can only evaluate to the following types: $\event{Complete}$ or $\event{Incomplete}$.
Since a malicious leader can stall the protocol by inactivity, the {\event{Incomplete}} type is used as a trigger for the leaders replacement.
If a \primitive{Disperse} invocation finishes with the type $\event{Complete}$, it will return a commitment $C$ and its associated combined signature containing $n-f$ out of $n$ replicas.
%
It's possible that a \event{Completed} \primitive{Disperse} invocation may not contain mini-blocks from all $n$ replicas. This is because either some honest replicas are late for the leader's collection time, or some percentage of malicious replicas are refusing to send mini-blocks. 
For the empty space, a protocol designer can meditate on the trade-offs and decide if that space should be used by the leader to fill in its own data; otherwise that empty space can be filled with zeroes indicating that the corresponding replicas are not included. We elaborate on this trade-off in Appendix~\ref{sec:appendix:card_design:lsc}.  
%
%
%
%
%
%
%

For a {\DAL} protocol, its censorship resistance depends on the characteristics of the {\primitive{Disperse}} invocation.
Definitions (\ref{def:replica-exclusion},~\ref{def:censorship-resistance}) quantify the amount of censorship resistance with a single parameter.
Given a desired parameter, the protocol has a large design space for choosing an inclusion mechanism for meeting the requirement.
Definition (\ref{def:accoutable-mechanism}) classifies the design space into two categories, based on the idea of whether a leader informs every replica about whose data are included. 
It is achieved with the use of attestation in definition (\ref{def:attestaion}).
If a design uses attestation in conforming the censorship resistance parameter, definitions (\ref{def:data-tampering},~\ref{def:space-capturing}) capture possible outcomes when a leader interferes with the inclusion mechanism. In Section~\ref{sec:card_design}, we provide protocols designed base on either of the two categories.
%

\begin{definition}[Replica-Exclusion]
\label{def:replica-exclusion}
    Suppose a \primitive{Disperse} invocation finishes with a \event{Complete} type with a valid combined signature, a replica $p$ is excluded if the leader does not include the mini-block from $p$ for the {\DAL} instance. Replica-exclusion happens even when the leader is honest.
\end{definition}

\begin{definition}[Censorship Resistance $t$] 
\label{def:censorship-resistance}
    A protocol is censorship resistant if it guarantees that at least $t$ honest replicas are not excluded in any {\DAL} instance. If $t=0$, the protocol is not censorship resistant. 
\end{definition}

\begin{definition}[Data-Attestation]
\label{def:attestaion}
    If a replica proposes some non-zero data as its mini-block for a {\DAL} instance, the replica declares it to a leader by generating a signature, called attestation, based on data in the mini-block. An attestation from replica $p$ is denoted $a_p$.
\end{definition}

\begin{definition}[Accountable-Inclusion-Mechanism]
\label{def:accoutable-mechanism}
    To conform with the censorship resistance $t$, an accountable inclusion mechanism requires a leader to explicitly demonstrate the required number of inclusion by showing attestations from other replicas, whereas a non-accountable inclusion mechanism allows the leader to demonstrate conformity without attestations. 
\end{definition}

%

\begin{definition}[Data-tampering]
\label{def:data-tampering}
    In an accountable mechanism, suppose a replica $p$ participates in a \primitive{Disperse} invocation with a non-zero mini-block and the leader used the attestation $a_p$ for demonstrating sufficient inclusion of mini-blocks to other replicas. Let $b_p$ 
    be the mini-block in the actual dispersed data corresponding to replica $p$, the data of replica $p$ is tampered if $a_p$ does not bind to $b_p$. 
\end{definition}
%

\begin{definition}[Space-capturing]
\label{def:space-capturing}
    Suppose in a \primitive{Disperse} invocation a replica $p$ is excluded from the {\DAL} instance, the data space designated for the replica $p$ is captured by a leader if in the actual dispersed data $b_p$, it is not set to a system default value reserved for the excluded replicas.
\end{definition}

The default value for the excluded replicas can be arbitrary, we let the value be identically $0$. To handle a corner case when a replica has no data to upload, the protocol defines a special NULL transaction for differentiating scenario.
%
%
With respect to definition(\ref{def:accoutable-mechanism}), a great difference between the two types of inclusion mechanism is shown in Lemma~\ref{lemma:accountability-diff}:
\begin{lemma}
\label{lemma:accountability-diff}
    Replica-Exclusion and Space-capturing in an accountable mechanism is publicly attributable; whereas in a non-accountable mechanism, both of them are only known privately by the leader and the excluded replica.
\end{lemma}
The proof is a direct consequence of using attestations. In a non-accountable mechanism, there is a lack of association between dispersed mini-blocks and their senders, so the leader can overwrite some space with its own data as long as it satisfies some requirement of censorship resistance $t$.
In contrast, an accountable inclusion mechanism keeps track of what mini-block is sent by which replica. 
See Appendix~\ref{sec:appendix:card_design:lsc} for more discussion.
%


When an accountable inclusion mechanism is used, it exposes an attack vector based on definition(\ref{def:data-tampering}), which does not exist for protocols that use non-accountable inclusion mechanisms.
The goal of the attack is to introduce inconsistency between what is being attested and what is actually dispersed.
An honest leader never tampers data;
whereas a malicious leader can tamper any mini-blocks including those produced by other malicious replicas.
With the definitions above, we describe properties required for a {\DAL} protocol.

\subsection{Properties of DA-CR protocols}
\label{sec:card_propery:properties}
\noindent A leader based {\DA} protocol that offers censorship resistance must provide the following properties for all {\DA} instances. The {\DA} protocol relies on partial synchronous assumption, different from asynchronous protocols including AVID\cite{cachin2005asynchronous}, AVID-M\cite{yang2022dispersedledger}.
The first four properties are important for integrating with offshelf BFT protocols, and the last three properties characterizes the censorship resistance of the integrated protocol.
\begin{itemize}
    \item \textbf{Termination}: If some replicas invoke \primitive{Disperse}$(B)$, then those replicas eventually \event{Finish} the dispersal.
    If at least $n-f$ replicas invoke \primitive{Disperse}$(B)$ and the leader is honest  with a network after GST, then the dispersal \event{Finishes} with type $\event{Complete}$.

    \item \textbf{Availability}: If an honest replica $\event{Finished}$ a dispersal with type $\event{Complete}$, then any honest retriever can invoke $\primitive{Retrieve}$ and eventually reconstruct a block $B'$. 
   
    \item \textbf{Correctness}: If an honest replica \event{Finished} a dispersal with a type $\event{Complete}$, then any honest retriever always retrieves the same block $B'$ by invoking \primitive{Retrieve}. If the leader is honest, then $B=B'$, where $B$ contains all the mini-blocks available to the honest leader at dispersal.

    \item \textbf{Commitment Binding}:  If an honest replica \event{Finished} a \primitive{Disperse}$(B)$ invocation with a type \event{Complete} and a valid combined signature on its commitment, no adversary can find another $B'$ that verifies correctly against the combined signature. 
    
    \item \textbf{Data-tampering Resistance}: In an accountable inclusion mechanism, if an honest replica \event{Finished} a \primitive{Disperse}$(B)$ invocation with type $\event{Complete}$,
    the dispersed data $B$ must ensure no mini-blocks are tampered in the {\DAL} instance.

    \item \textbf{Inclusion-\textit{t}}: If an honest replica \event{Finished} a \primitive{Disperse}$(B)$ invocation with type \event{Complete}, 
    the dispersed data must include at least $t$ mini-blocks from honest replicas.

    \item \textbf{Space-Capturing}-$\eta$: If an honest replica \event{Finished} a \primitive{Disperse}$(B)$ invocation with type \event{Complete},
    the dispersed data contain at most $\eta$ captured mini-blocks by a leader.
    
\end{itemize}

We elaborates differences among those properties in Appendix~\ref{sec:appendix:card_property:discussion}.
Table~\ref{table:card_comparison} summarizes the parameters of {\DA} properties offered by the family of {\name} protocols.
For ease of reference, we use the suffix of each protocol combining with {\card} as the name for {\DA} component.
%
{\card} stands for \textbf{c}ensorship \textbf{r}esistance \textbf{d}ata \textbf{a}vailability with flexible ordering.
%
%

\begin{table}[!hbt]
\caption{Comparison of {\DA} component of BFT protocols}
\begin{subtable}[h]{\linewidth}
\resizebox{\textwidth}{!}{
\begin{tabular}{|c | c | c | c | c | c | }
\hline
\multirow{2}{*}{}  {\name} & DA-CR & BFT $(n\ge)$ & $t$ & $\eta$  & message \\
                     protocol & protocol  & assumption &  &  & complexity \\
\hline \hline
{\nameOne} & {\cardOne} & $3f+1$     & ${n-2f}$ & $0 {\dagger}$ & $O(b + \lambda n) $ \\
\hline
{\nameQuarter}& {\cardQuarter}  & $4f+1$ & $n-3f$ & $n-f-1$ & $O(b + \lambda)$ \\
\hline
{\nameLite} & {\cardLite} & $3f+1$ & $n-2f-o(f / \log(f))$ & $o(f/\log{f}){\dagger}$ & $O(b + \lambda \kappa \log n)$ \\
\hline
\end{tabular}
}
\caption*{$\lambda$ is a security parameter with fixed size, which uniquely represents the data block. $\kappa$ is a small constant. ${\dagger}$ means the parameter $\eta$ can increase up to $\eta=f$, see Appendix~\ref{sec:appendix:card_design:lsc}.}
\end{subtable}
\label{table:card_comparison}
\end{table}

\section{DA-CR Protocol Designs}
\label{sec:card_design}

Section~\ref{sec:card_property} defines the properties required from a
censorship-resistant data availability layer. This section gives a concrete
construction that realizes these properties while preserving the central design
goal of $\name$: data remains sharded throughout the consensus critical path.

A DA-CR construction must constrain the leader without forcing every validator
to reconstruct the full data object. The leader should not be able to
arbitrarily remove honest mini-blocks, replace them with its own data, or
obtain a certificate for unavailable data. A fully replicated design could
enforce these checks by asking every validator to download the entire payload,
but doing so would eliminate the scalability benefit of sharding.

Instead, ordinary validators in $\name$ vote using only their assigned shards,
compact commitment evidence, and the DA-CR inclusion condition. The challenge
is to make these local checks imply a global guarantee: the certified data is
available, bound to the committed object, and captures enough
replica-contributed mini-blocks. The construction below achieves this by
combining erasure coding with KZG commitments. The erasure code keeps data
available in sharded form, while the linear commitment relation binds the
leader to a single global data object.

\subsection{Data Layout and KZG Linearity}
\label{sec:card_design:data_layout}

The construction keeps data sharded by combining a positional data layout with
the linearity of Reed--Solomon encoding and KZG commitments, see Figure~\ref{fig:card_data_struct}. For a DA-CR
instance, replica $p_i$ contributes a mini-block, or blob, denoted $b_i$. The
corresponding KZG commitment is denoted $c_i=\mathsf{KZG.Commit}(b_i)$. The
leader assembles the replica contributions into an ordered vector
$B=(b_1,\ldots,b_n)$ and the corresponding commitment vector
$\mathbf{c}=(c_1,\ldots,c_n)$, where position $i$ is reserved for replica
$p_i$.

This positional layout is part of the protocol. It lets the DA-CR layer reason
about which replica's contribution is included, which replica is excluded, and
which space is controlled by the leader. Each position has one of three
meanings. If $p_i$ contributes a non-empty mini-block, then position $i$
contains $b_i$ and, in an accountable construction, the leader must provide the
corresponding attestation. If $p_i$ participates but has no transaction data, it
contributes a distinguished NULL mini-block. If $p_i$ is excluded from the
DA-CR instance, then position $i$ is filled with the system default value,
written as $0$, unless the construction explicitly allows the leader to capture
that space. The distinction between NULL and $0$ separates honest empty
participation from leader exclusion.

After forming $B$, the leader applies Reed--Solomon encoding to obtain
$\widehat{B}=\mathsf{Encode}(B)$. Because Reed--Solomon encoding is linear,
each encoded shard is a linear combination of the original mini-blocks. Because
KZG commitments are also linear, the commitment to an encoded shard can be
derived from the same linear combination of the commitments
$c_1,\ldots,c_n$. Concretely, if an encoded shard is
$\widehat{b}_j=\sum_{i=1}^{n}\alpha_{j,i}b_i$, then its commitment is checked
against $\widehat{c}_j=\prod_{i=1}^{n}c_i^{\alpha_{j,i}}$, using additive or
multiplicative notation depending on the commitment group\cite{nazirkhanova2021information}.

This is the mechanism that makes local validation meaningful. Validator $p_j$
does not download the full vector $B$ or reconstruct the full encoded object
$\widehat{B}$. Instead, it receives the systematic commitment list
$\mathcal{C}=(c_1,\ldots,c_n)$, the DA-CR inclusion evidence, and its assigned
encoded chunks. The validator locally extends $\mathcal{C}$ to
$(c_1,\ldots,c_{3n})$ using Reed--Solomon linearity, then checks that each
received chunk commits to the corresponding derived commitment. The check is
therefore local in data but global in commitments: each validator verifies only
a few chunks, while the commitment list binds those chunks to one globally
committed layout of replica mini-blocks.

The leader therefore cannot freely replace an honest mini-block while
preserving the commitment relation, unless the protocol allows that position to
be excluded or captured under the DA-CR rule. Inclusion, tamper-resistance, and
space-capturing are expressed over the global layout of $B$, while validators
remain on a sharded critical path.

\subsection{{\cardOne} Protocol}
\label{sec:card_design:cardOne}

We now instantiate the layout above as a concrete DA-CR protocol. The
{\cardOne} protocol uses Reed--Solomon encoding, KZG commitments, and an
aggregate multi-signature scheme to implement a \primitive{Disperse} invocation
with parameters $(t=n-2f,\eta=0)$. It is an accountable construction:
the leader must show enough replica attestations, and the committed data must
remain consistent with the attested mini-blocks.

\begin{figure}[h]
\centering
\begin{subfigure}{0.5\textwidth}
\centering
\includegraphics[width=\textwidth]{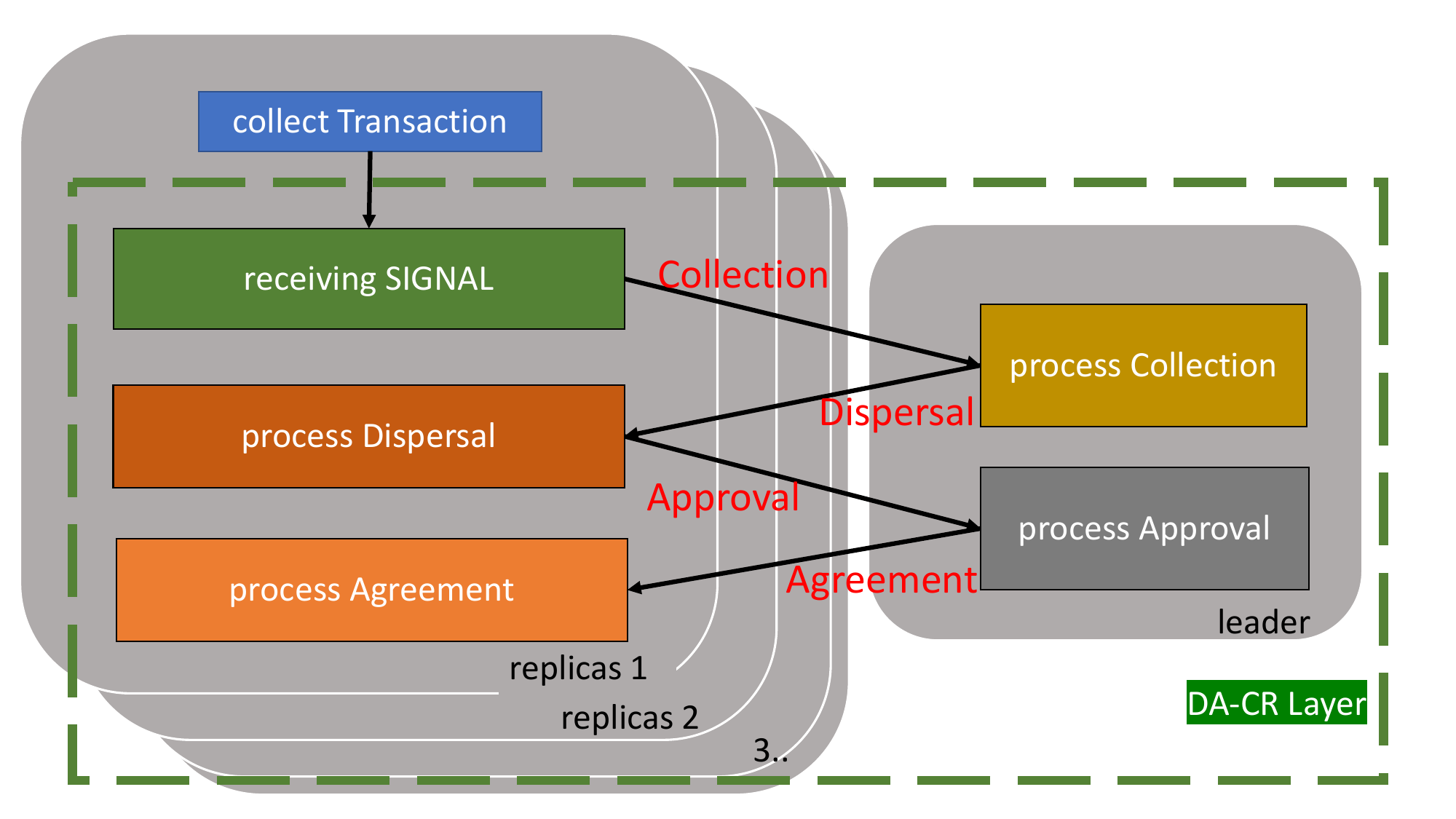}
\caption{Flow diagram for a {\DAL} protocol } 
\label{fig:nav_and_da}    
\end{subfigure}%
\begin{subfigure}{.5\textwidth}
\centering
\includegraphics[width=0.9\textwidth]{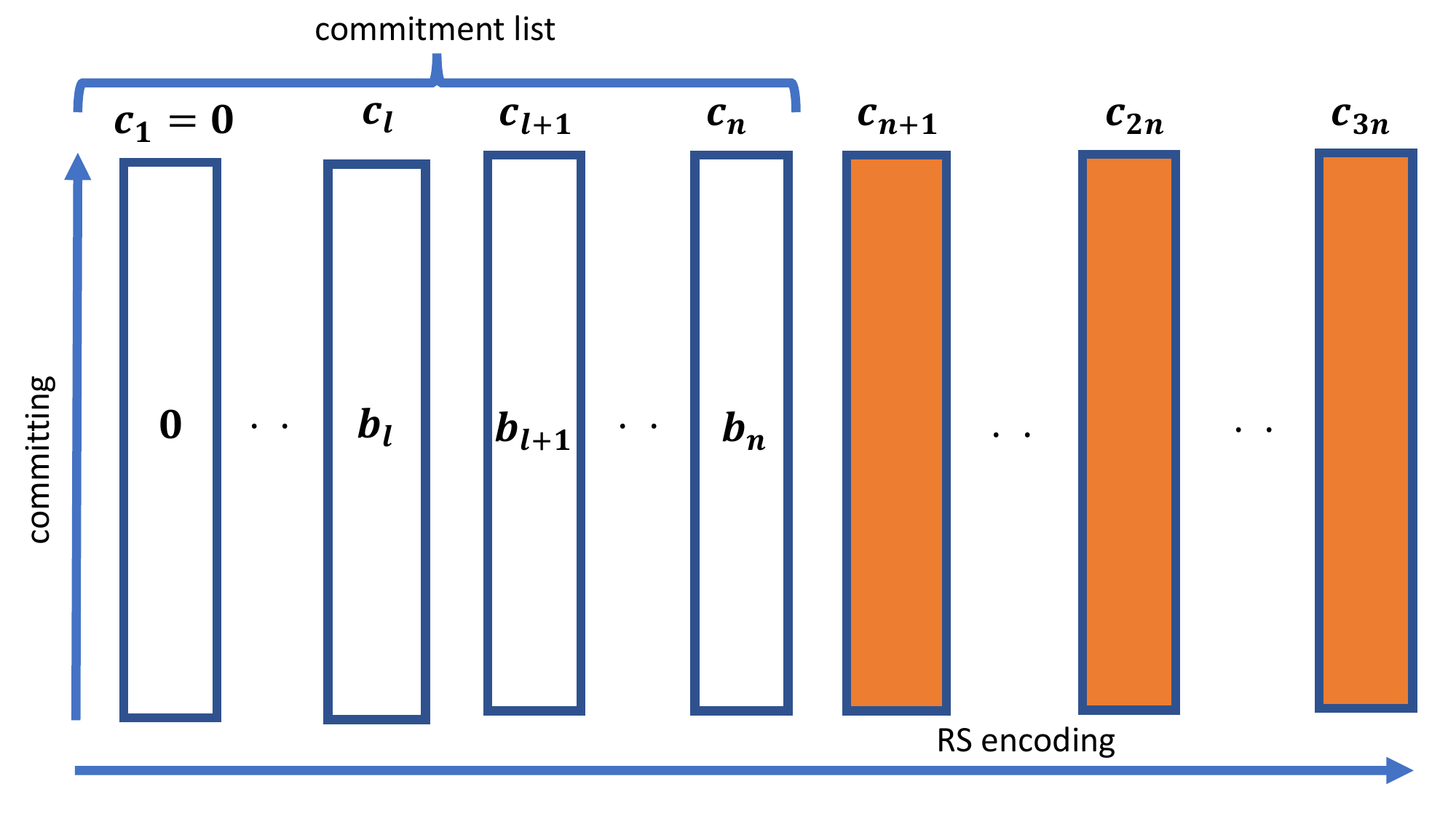}
\caption{A 2D matrix data structure used by the leader for producing $3n$ commitment. In the diagram, replica $1$ is excluded, as its data is $0$. $b$ denotes mini-block(systematic chunk), orange are for parity chunks.} 
\label{fig:card_data_struct}
\end{subfigure}
\end{figure}

At a high level, \primitive{Disperse} consists of four messages, see Figure~\ref{fig:nav_and_da}. In the
\textsf{Collection} phase, replicas send mini-blocks and attestations to the
leader. In the \textsf{Dispersal} phase, the leader encodes the collected
mini-blocks, commits to the encoded object, and sends each replica its assigned
shards with the corresponding commitment evidence. In the \textsf{Approval}
phase, replicas verify their local shards and the inclusion evidence, then sign
the resulting DA-CR commitment. In the \textsf{Agreement} phase, the leader
aggregates $n-f$ approvals into a certificate for the completed DA-CR instance.

The invocation is triggered by the surrounding BFT protocol. We model this
trigger as a \textsf{Signal} message carrying the current view number. The
details of how this signal is produced by the BFT pacemaker are orthogonal to
the DA-CR construction and are discussed in
Section~\ref{sec:bigDipper_hotstuff}. For this section, it is enough to assume
that all honest replicas start the DA-CR instance for the same view.

Before giving the protocol, we fix the notation used in
Algorithm~\ref{alg:card}. The variable $v$ denotes the current BFT view and also
identifies the corresponding {\DAL} instance. Replica $p_i$'s mini-block in
view $v$ is $b_i^v$, its KZG commitment is
$c_i^v=\mathsf{kzgCommit}(b_i^v)$, and its attestation is
$a_i^v=\mathsf{sign}(c_i^v\|v)$. The leader collects such attestations into a
set $S$, also denoted {\al}. In the dispersal message, the first $n$
commitments $(c_1,\ldots,c_n)$ correspond to the systematic positions: if
$p_i$'s mini-block is included then $c_i=\mathsf{kzgCommit}(b_i^v)$, and if
$p_i$ is missing or excluded then $c_i=0$.

Each replica receives $(c_1,\ldots,c_n)$, the attestation set $S$, and its
assigned parity chunks. It locally extends $(c_1,\ldots,c_n)$ to
$(c_1,\ldots,c_{3n})$ using the Reed--Solomon linearity, and signs the digest
$C=\mathsf{CRHF}(c_1,\ldots,c_{3n})$ if all checks pass. A completed {\DAL}
instance returns $C$ together with an aggregate signature from at least $n-f$
replicas.

The four-message flow is shown in Algorithm~\ref{alg:card}. The algorithm is
written from the perspective of both a regular replica and the leader. A regular
replica starts a {\DAL} instance after receiving a \textsf{Signal}, sends its
mini-block and attestation in a \textsf{Collection} message, verifies the
leader's \textsf{Dispersal} message, and signs the resulting commitment. The
leader collects mini-blocks, encodes them, disperses shards, and aggregates
$n-f$ approvals into an \textsf{Agreement} certificate.

\begin{algorithm}[!htbp]
    {\fontsize{7.5pt}{7.5pt}\selectfont \caption{{\cardOne}}
    \label{alg:card}
    \begin{algorithmic}[1]
        \State \underline{\textbf{At the replica } $p$ :}
        \State $\textit{myView} \gets 0$, $\textit{aggSigs}[\cdot] \gets \emptyset$, $\textit{db} \gets \emptyset$
        \Event{$ \textit{receiving SIGNAL}(v)$} \label{alg:card:line:signal_start}
            \State $\textit{myView} \gets v$
            \State $ b_p^v \gets mempool$
            \State $c_p^v \gets kzgCommit(b_p^v)$
            \State $a_p^v \gets sign(c_p^v|| v)$
            \State $\text{send } \langle \textsf{Collection}, a_p^v, b_p^v\rangle \text{ to leader}$
        \EndEvent \label{alg:card:line:signal_end}
        \Event{$ \textit{receiving} \textsf{ Dispersal}$}
            \State $\{c_i\}_{1\le i \le 3n} \gets extendEval(\{c_i\}_{1\le i \le n})$ \label{alg:card:line:dipserse_start}
            \State {$a \gets (checkCom(c_{p}, b_p^v) \lor c_p =0)\wedge checkInc(S, \{c_i\}_{1\le i \le n}) $} \label{alg:card:line:sig_check} \label{alg:card:line:check1}\\
            \hspace*{8mm} {$ \wedge \text{ } checkCom(c_{p+n}, chunk_{p+n})$ $\wedge$ $ checkCom(c_{p+2n}, chunk_{p+2n})$}  \label{alg:card:line:dipserse_check_end} \label{alg:card:line:check2}
            \If{$a=\true$}
                \State $db[c_p,c_{p+n},c_{p+2n}] \gets (b_p^v, chunk_{p+n}, chunk_{p+2n})$
                \State $C \gets CRHF(\{c_i\}_{1 \le i \le 3n})$
                \State $\sigma_{p} \gets  \textit{ms.sign}(sk_p,C)$
                \State $\text{send } \langle \textsf{Approval}, \sigma_{p}\rangle$ to leader
            \EndIf
        \EndEvent \label{alg:card:line:dipserse_end}
        \Event{$ \textit{receiving } \textsf{Agreement}$}
            \If{$\textit{ms.verify}(m.\sigma_{agg}, m.I_{\sigma_{agg}}, m.C)$} 
                \State $aggSigs[m.g, m.v] \gets (m.\sigma_{agg}, m.I_{\sigma_{agg}})$
            \EndIf
        \EndEvent
        \State \underline{\textbf{At the leader:}}
        \Event{$ \textit{receiving } \textsf{Collection} \wedge \textit{ validCollect(m) } \wedge \textit{reach-proc-time}   $} \label{alg:card:line:collection_start}
            \State $B \gets \{ m.b_p^v | m.a_p^r \in S \} $
            \State $\{c_i, chunk_i\}_{1 \le i \le 3n}\gets encode2d(B)$
            \For{$i \gets 1 \cdots n $}
                \State {$\text{send i } \langle \textsf{Dispersal}, \{c_i\}_{1\le i \le n}, chunk_{i+n}, chunk_{i+2n}, S \rangle$}
            \EndFor       
        \EndEvent  \label{alg:card:line:collection_end}
        \Event{$\textit{receving }\textsf{Approval} \wedge validApproval(m) $} \label{alg:card:line:aggregate_start}
            \State $\sigma_{agg}, I_{\sigma_{agg}} = \text{ms.agg}( \{m.\sigma_{i_p}, pk_{i_p}\}_{1 \le p \le n-f})$
            \For{\texttt{i = 1..n}}
                \State {$\text{send replica i } \langle \textsf{Agreement}, I_{\sigma_{agg}}, \sigma_{agg}, C \rangle$}
            \EndFor
        \EndEvent  \label{alg:card:line:aggregate_end}
    \end{algorithmic}}
\end{algorithm}

\begin{algorithm}[!htbp]
    {\fontsize{7.5pt}{7.5pt}\selectfont \caption{utility functions}
    \label{alg:utility}
    \begin{algorithmic}[1]
        \Procedure{$\text{checkCom}$}{ $commitment$, $chunk$}
            \State \Return $kzgCommit(chunk) = commitment$
        \EndProcedure
        \Procedure{$\text{extendEval}$}{ $\{c_i\}_{1\le i \le n}$}
            \State \Return $rsEncode(\{c_i\}_{1\le i \le n}, 1/3)$ $  \textit{\# 1d RS to 3n with systematic encoding}$
        \EndProcedure
        \Procedure{$\text{checkInc}$}{ $S$, $\{c_i\}_{1\le i \le n}$} $\quad \textit{S is the attestation set}$ \label{alg:utility:line:checkInc} 
            \State $O \gets \emptyset$, $A \gets \{1 \cdots n\}$
            \For{$a_p^v \textit{ in } S$}
                \If{$\textit{verifySig}(a_p^v || myView, c_p) = \true \wedge c_p \neq 0 $}
                    \State $O \gets O \cup p$
                \Else
                    \State \Return $\false$
                \EndIf
            \EndFor
            \For{$p \gets A - O$} \label{alg:utility:line:spaceCapture}
                \If{$c_p \neq 0$}
                    \State \Return $\false$
                \EndIf
            \EndFor
            \State \Return $|O| \ge n - f$ 
        \EndProcedure 
        \Procedure{$\text{validCollect}$}{ $m$}
            \If{$\textit{verifySig}(msg.a_p^v || myView, \text{kzgCommit}(m.b_p^v)) $}
                \State $S \gets S \cup m.a_p^v$  
            \EndIf
            \If{$len(S) \ge n-f$} 
                \State \Return $\true, S$
            \EndIf
        \EndProcedure
        \Procedure{$\text{validApproval}$}{ $m$}
            \If{$\textit{ms.verify}(\sigma_p, pk_p) $}
                \State $M \gets M \cup (m, pk_p)$
            \EndIf
            \If{$len(M) \ge n-f$} 
                \State \Return $\true, M$
            \EndIf
        \EndProcedure
        \Procedure{$\text{encode2d}$}{ $B$}
            \State $\{chunk_i\}_{1 \le i \le 3n} \gets \text{rsEncode2d}(B, 1/3)$
            \For{$i \gets 1 \cdots 3n$}
                \State $c_i \gets \text{kzgCommit}(chunk_i)$
            \EndFor
            \State\Return $\{c_i, chunk_i\}_{1 \le i \le 3n}$
        \EndProcedure
    \end{algorithmic}}
\end{algorithm}  

\subsubsection{Why the Checks Give $(t=n-2f,\eta=0)$}

The checks in Algorithm~\ref{alg:card} enforce the parameters claimed for
{\cardOne}. The function \textsc{checkInc} requires the leader to present
$n-f$ valid attestations for non-zero systematic commitments. Since at most $f$
of these attestations can come from Byzantine replicas, at least $n-2f$
included mini-blocks are from honest replicas, giving $t=n-2f$.

The same check prevents space capturing. For any systematic position without a
valid attestation, \textsc{checkInc} requires the corresponding commitment to
be zero. Hence the leader cannot use an excluded replica's position for its own
data, giving $\eta=0$.

Finally, data tampering is prevented by the combination of attestations and KZG
checks. If the leader uses replica $p_i$'s attestation, then the systematic
commitment $c_i$ must match the mini-block attested by $p_i$. The parity
commitments are derived from the systematic commitments using the linearity of
Reed--Solomon encoding and KZG commitments, and each approving replica verifies
the chunks it receives against those derived commitments. Therefore the leader
cannot change an attested mini-block or its encoded shards while still passing
validation, except by breaking commitment binding. 

The remaining properties---termination, availability, correctness, and
commitment binding---follow from the dispersal certificate, erasure-coded
retrievability, and KZG binding; full proofs appear in
Appendix~\ref{sec:appendix:card7_proof}.

\subsubsection{\primitive{Retrieve} Invocation}
\label{sec:card_design:retrieve}

The \primitive{Retrieve} invocation is formally defined in
Algorithm~\ref{alg:retrieve} in Appendix~\ref{sec:appendix:retrieve_algorithm}.
It is the read interface for data certified by a completed {\DA} instance.
Retrieval may be invoked by clients, replicas, or upper-layer services that need
full or partial data. It is not part of the consensus voting path: validators
approve a dispersal by checking their local chunks, the commitment list, and
the DA-CR inclusion evidence during \primitive{Disperse}.

In the common case, a retriever requests the desired data from the leader or
from the replica that stores the corresponding chunk. For a full mini-block,
correctness is checked by recomputing the KZG commitment and comparing it with
the corresponding entry in {\cl}. For partial data, the retriever requests the
data and witness and verifies them against the same commitment. If the leader
or responsible replica does not respond, the retriever queries other replicas,
collects chunks consistent with the commitments, and runs \textit{rsDecode} to
recover the desired data.

The retrieved data is then consumed by the upper layer. An application may run
a derivation and execution bundle, exposed through RPC endpoints, that downloads
certified data, verifies it against {\cl}, derives the application payload, and
serves results to users. This execution role may be implemented by full-data
executors, an application committee, TEEs, or software whose outputs are checked
by zk or optimistic proofs. These choices are outside the {\DA} protocol. The
{\DA} layer guarantees availability and commitment binding; the upper layer
decides how certified data is executed and served.

Because each replica stores three chunks in {\cardOne}, a retriever needs to
contact a fraction of replicas to reconstruct data with high probability.
Appendix~\ref{sec:appendix:retrieve:num_reconstruct} gives the analysis.

\subsection{{\cardQuarter}: A Non-Accountable Variant}
\label{sec:card_quarter}

{\cardQuarter} explores the low-overhead end of the DA-CR design space. Unlike
{\cardOne}, where every replica receives the commitment list {\cl} and the
attestation list {\al}, {\cardQuarter} removes accountable inclusion entirely.
Replicas send mini-blocks to the leader without attestations; after collecting
at least $q\ge n-f$ mini-blocks, the leader concatenates the collected data,
encodes it with a single Reed--Solomon polynomial, and commits to it with one
KZG commitment. During dispersal, the leader sends coded data only to the
replicas whose mini-blocks were collected, so each dispersal message contains
only constant-size commitment evidence and coded data.

The cost is weaker public accountability. In {\cardOne}, even an excluded
replica can verify the attestation list and approve that enough replica data was
included. In {\cardQuarter}, an excluded replica has no public evidence linking
included mini-blocks to their original senders, so only included replicas can
safely approve. The leader may therefore have to complete with only $n-2f$
approvals: up to $f$ honest replicas may be excluded, and up to $f$ Byzantine
replicas may refuse to approve. Since as many as $f$ approvers can be
Byzantine, reconstruction must succeed from at least $n-3f$ honest approvers,
which requires the stronger assumption $n\ge 4f+1$. Thus {\cardQuarter}
achieves constant per-replica dispersal overhead, but with weaker inclusion,
larger residual leader influence, and a stronger liveness threshold, as
summarized in Table~\ref{table:card_comparison}. Full details and proofs appear
in Appendix~\ref{sec:appendix:non-accountables:proofs_cardQuarter}.


\subsection{{\cardLite}: Sampling-Based Accountable Inclusion}
\label{sec:card_lite}

{\cardLite} is an optimization of {\cardOne} that keeps accountable inclusion
and the standard $n\ge 3f+1$ assumption, while avoiding the need to send the
full commitment list {\cl} and attestation list {\al} to every replica. The
construction is DAS-like in spirit: instead of having every replica verify the
entire accountable structure, each replica checks a small number of sampled
positions. The goal is to preserve the key guarantee of {\cardOne} with only $O(\log n)$
per-replica overhead.

The main cryptographic tool is a doubly homomorphic commitment
scheme~\cite{bunz2021proofs}. In {\cardLite}, the entries of {\cl} are KZG
commitments, and the leader commits to this commitment list itself. This
``commitment to commitments'' lets the leader open a sampled position with a
logarithmic-size witness with pedersen commitment scheme: for a sampled replica position $i$, the leader reveals
the attestation $a_i$, the mini-block commitment $c_i$, and a proof that $c_i$
is indeed the $i$-th entry of the committed {\cl}. Thus each sample checks both
accountability and consistency with the underlying 2D KZG structure.

Sampling is used to distinguish an honest leader from a censoring leader. An
honest leader that includes enough mini-blocks can answer sampled challenges;
a leader that omits many mini-blocks loses the ability to answer too many of
them, because the corresponding attestations or commitment openings are
missing. To avoid an interactive challenge protocol, {\cardLite} derives the
sampled positions from the final commitment using the Fiat--Shamir transform.
The leader sends each replica only one valid sampled batch, and the replica
recomputes the challenges and verifies the batch locally.

With appropriate parameters, this sampling test detects a leader that censors
$\Omega(f/\log f)$ mini-blocks with overwhelming probability, while an honest
leader with up to $f$ missing mini-blocks caused by Byzantine inactivity still
passes. Therefore {\cardLite} achieves
$t=n-2f-o(f/\log f)$ with $O(\log n)$ overhead, as summarized in
Table~\ref{table:card_comparison}. The full protocol, parameter choice, and
analysis are deferred to Appendix~\ref{sec:appendix:card_lite}.

\section{Transaction Broadcast}
\label{sec:navigation}

The $\DAL$ layer guarantees that a completed instance captures enough
replica-contributed mini-blocks. Transaction broadcast determines whether a
client's transaction reaches those mini-blocks before they are submitted. This
section describes {\navigation}, a simple broadcast protocol in which a client
chooses a fanout parameter $x$ and sends the transaction to $x$ available
replicas for a target view $v$.

The fanout parameter exposes the short-term censorship-resistance/cost
tradeoff at the client layer. Larger $x$ gives stronger prompt-inclusion
probability but consumes more data-availability bandwidth; smaller $x$ is
cheaper but accepts higher censorship risk. This avoids forcing every
transaction to pay for maximum censorship resistance, which would recreate a
form of full replication at the submission layer. The rest of this section
defines {\navigation} and analyzes deterministic inclusion, probabilistic
inclusion, and broad-censorship bounds.

\subsection{{\navigation} Protocol}
\label{sec:navigation:protocol}

The {\navigation} protocol is defined in Algorithm~\ref{alg:navigation}. A
client submits a transaction by choosing three inputs: the transaction $tx$, a
fanout parameter $x$, and a target view $v$. The fanout $x$ is the number of
distinct replicas the client wants to contact. The target view $v$ identifies
the {\DAL} instance in which the client wants prompt inclusion.

The client first calls \textsc{getAvailableReplicas} to obtain replicas that
can still include transactions in their mini-blocks for view $v$. A concrete
implementation of this procedure is given in
Appendix~\ref{sec:appendix:navigator:tx_broadcast_sync}. If at least $x$
available replicas are returned, the client sends $\langle tx,v\rangle$ to
each of them. Each replica responds with $1$ if it can still place the
transaction into its mini-block for view $v$, and $0$ otherwise. The number of
positive responses is returned to the client and can be used as the effective
fanout in the inclusion-probability analysis.

The view number $v$ should be interpreted as a prompt-inclusion target, not as
an additional consensus safety assumption. A client that cares about
short-term censorship resistance can pre-schedule a transaction by sending it
for a future view. An honest replica that receives $\langle tx,v\rangle$
before forming its mini-block for view $v$ stores the transaction with that
view tag and includes it when the corresponding {\DAL} instance starts. If the
transaction arrives too late, the replica returns failure and the client can
retry for a later view. Thus client timing affects the achieved fanout, and
therefore the inclusion probability for the target view, but it does not affect
BFT safety or the {\DAL} validity rule.

The target view \(v\) in {\navigation} should be interpreted as a
prompt-inclusion target, not as an additional consensus safety assumption. A
transaction receives the prompt-inclusion guarantee for view \(v\) only if it
reaches honest replicas before they form their mini-blocks for that view. This
is analogous to slot-cutoff assumptions in multi-proposer protocols, where
transactions must reach honest proposers before the relevant slot deadline to
receive that slot's inclusion guarantee~\cite{mcp}. If a transaction arrives
too late, the replica returns failure and the client can retry for a later view;
this affects the achieved fanout and inclusion probability, but not BFT safety
or the {\DAL} validity rule.

\begin{algorithm}[!t]
    {\fontsize{7.5pt}{7.5pt}\selectfont \caption{\navigation}
    \label{alg:navigation}
    \begin{algorithmic}[1]
        \State \underline{\textbf{At the client:} }
        \Procedure{$submit$}{ $x$, $tx$, $v$}
            \State $S  \gets getAvailableReplicas(x, v)$, $cn \gets 0$ \label{alg:navigation:line:get_avail}
            \If{$ |S| < x$}
                \State\Return $FAIL, |S|$
            \EndIf
            \For{$j \text{ in } S $}
                \State $r \gets \text{send } \langle tx,v \rangle \text{ to replica } j$
                \State $cn \gets cn + r$
            \EndFor
            \State\Return $cn$
        \EndProcedure
        \State \underline{\textbf{At the replica } $p$ :}
        \Event{$ \textit{receiving } \langle tx, v \rangle $}
            \If{$(v > \textit{myView}) \lor (v = \textit{myView } \wedge \text{waiting for Signal})$}
                \State $\text{insert } \langle tx, v \rangle \text{ to mempool}$ \textit{ and respond 1}
            \Else
                \State $\textit{respond 0}$
            \EndIf
        \EndEvent
    \end{algorithmic}}
\end{algorithm}

\subsection{Inclusion Guarantees}
\label{sec:navigation:inclusion_guarantees}

The inclusion probability of a transaction is determined by two quantities:
the fanout \(x\), i.e., the number of distinct available replicas contacted by
the client, and the number \(q\) of honest mini-blocks captured by the completed
\({\DAL}\) instance. The value of \(q\) is always at least the
\property{Inclusion-\(t\)} parameter of the \({\DAL}\) protocol. When the
leader is honest, \(q\) may be larger, up to \(n-f\), depending on network
timing and the collection waiting time. In the worst case, a malicious leader
may capture only the minimum \(t\) honest mini-blocks allowed by the \({\DAL}\)
rule.

\begin{table}[!hbt]
\caption{Probability of inclusion as a function of client fanout and honest mini-block inclusion.}
\begin{subtable}[h]{\linewidth}
\centering
\resizebox{0.6\textwidth}{!}{
\begin{tabular}{|c | c | c  | c|}
\hline
Leader type          &  Malicious    & Honest & Honest \\
\hline
number honest mini-block       &  $q=t$        & $q=t$  & $q=n-f$\\
\hline \hline
$n-t+1 \le x \le n$       &  $\Pr(IN)=1$  & $\Pr(IN)=1$ & $\Pr(IN)=1$\\
\hline
$f + 1 \le  x <n-t+1$    &  see~\ref{sec:bigdipper_prob_bad_leader}  & see~\ref{sec:navigation:honest_leader} & $\Pr(IN)=1$\\
\hline
$0 <  x <f+1$    &  {$\Pr(IN)=0$}  & see~\ref{sec:navigation:honest_leader} & see~\ref{sec:navigation:honest_leader} \\
\hline
\end{tabular}
}
\caption*{\(x\) is the number of transaction copies sent to distinct available replicas. \(q\) is the number of honest mini-blocks captured by the completed \({\DAL}\) instance.}
\end{subtable}
\label{table:prob_inclusion_table}
\end{table}

Table~\ref{table:prob_inclusion_table} summarizes the main regimes. It is
computed under the conservative assumption that all \(f\) Byzantine replicas
actively censor the transaction. The table separates deterministic cases, where
inclusion follows by counting, from probabilistic cases analyzed in the next
subsection.
The deterministic entries follow from the pigeonhole principle. If a client
contacts \(x\) replicas and the completed \({\DAL}\) instance captures \(q\)
honest mini-blocks, then inclusion is guaranteed whenever \(x+q\ge n+1\).
In particular, because every completed \({\DAL}\) instance captures at least
\(t\) honest mini-blocks, sending to \(n-t+1\) available replicas guarantees
inclusion.

The opposite extreme occurs under a malicious leader when \(x\le f\). In that
case, a targeted transaction may be delivered only to replicas whose
mini-blocks the leader can avoid, so the transaction can be censored in a
single instance. For the remaining non-binary cases, the exact inclusion
probability depends on the pair \((x,q)\). Section~\ref{sec:navigation:honest_leader}
derives the probability under an honest leader, Section~\ref{sec:bigdipper_prob_bad_leader}
handles the malicious-leader case, and Section~\ref{sec:bigdipper_upper_bound}
bounds broad censorship across many independently submitted transactions.

\subsection{Probabilistic Transaction Inclusion}
\subsubsection{Probability of Inclusion with an Honest Leader}
\label{sec:navigation:honest_leader}

Given the leader is honest
%
and the client uses Algorithm~\ref{alg:navigation} to choose $x$ replicas randomly without replacement,
let $X$ be a random variable indicating the number of honest replicas that have the client's transaction.
For $0<i \le \min{(x, n-f)}$, the probability of the transaction being included by $X=i$ honest replicas is 
$\Pr(X=i;x) = {n-f \choose i}{f \choose x-i} / {n \choose x}$.


%
%
%
Assume $f$ malicious mini-blocks has $0$ network latency, and are collected by the honest leader, the number of collected honest mini-block $q \ge n-2f$, because the leaders need to at minimal have $n-f$ mini-blocks.
%
The probability of inclusion by any honest mini-block is
%
%
%
%
\begin{align*}
    \centering
    \label{eq:1}
    \Pr(\textit{IN; x,q} )   
             &= 1 - \sum_{i=1}^x \Pr(\textit{None } | X=i; q ) \Pr{(X=i;x)} \tag{1}
\end{align*}
where $\Pr(\textit{None } | X; q)$ is the probability none of $X$ honest mini-block containing the client's transaction. If $X + q > n$, the probability of exclusion $\Pr(\textit{None } | X;q) = 0$. 
If all honest replicas have independent and identically distributed network latency, we can explicitly calculate the conditional probability as 
$\Pr(\textit{None } | X=i;q ) = \prod_{j=0}^{q-1} \frac{n-f-i-j}{n-f-j}$, and the probability of exclusion decreases exponentially as $q$ increases.
%
%
%
In most cases, a client only needs to send a few transactions to achieve high probability of inclusion. We provide an empirical evaluation in Appendix~\ref{sec:appendix:navigator:empirical_eval_with_honest_leader} and presents an entire table of probability in Figure~\ref{fig:heatmap_honest_leader}.

\subsubsection{Probability of Inclusion with a Malicious Leader}
\label{sec:bigdipper_prob_bad_leader}
When a leader is malicious, but a client sends $x \ge f+1$ copies of a transaction, 
%
%
%
although some of them might be received by malicious replicas, the chance of $f+1$ honest replicas receiving it increases with $x$.
Since the malicious leader must include $f+1$ honest mini-block, 
%
%
the probability of inclusion equals to the probability that the client delivers the transaction to at least $f+1$ honest replicas, which can be computed $\Pr(IN) = \sum_{i=f+1}^{x}({n-f \choose i} {f \choose  x-i} / {n \choose x})$, where $f+1 \le x < n-t+1$. 
Further analysis is provided in Appendix~\ref{sec:appendix:nagivator:malicious_leader_f1}. 

\subsection{Upper Bound on Total Censored transactions}
\label{sec:bigdipper_upper_bound}

A malicious leader may be able to censor a targeted transaction when the client
uses small fanout. However, censoring many independently submitted transactions
is harder. Suppose $T$ clients each send $x$ copies to random distinct replicas.
Let $p$ be the probability that a transaction reaches at least one replica that
the adversary cannot avoid while satisfying the DA-CR inclusion rule. Under the
sampling model used in Appendix~\ref{sec:appendix:nagivator:malicious_leader_upper},
we have
$
    p = 1- \prod_{i=0}^{x-1} \frac{n-f}{n-i}.
$
Using Hoeffding's inequality and a union bound over adversarial choices of
captured replicas, the probability that the adversary censors an $\alpha$
fraction of the $T$ transactions is at most
\[
    2 {n \choose n-f} \exp\{-2(p+\alpha-1)^2T\}.
\]
Thus, for fixed $x$ and $\alpha$, the probability of broad censorship decreases
exponentially in the number of submitted transactions.

\subsection{Erasure Coded {\navigation}}

{\navigation} should be viewed as a minimal submission interface rather than
the only possible dissemination strategy. It sends full transaction copies to a
chosen set of available replicas, which makes the cost/security tradeoff easy
to analyze. More sophisticated protocols can be built at the same boundary. For
example, Sedna shards a transaction into coded symbols and disseminates those
symbols across proposer lanes, addressing the user-side dissemination problem
left open by MCP-style protocols: full replication improves censorship
resistance but hurts goodput and exposes payloads, while sending to few
proposers gives weaker censorship and latency guarantees~\cite{sedna}. A
Sedna-style strategy could replace full-copy submission in {\navigation},
improving bandwidth efficiency and providing transaction hiding until enough
symbols are reconstructed. The DA-CR and BFT layers are orthogonal to this
choice; they only require that enough replica-contributed data is captured by
the completed instance.

\section{{\name} with BFT Integration}
\label{sec:bigDipper_hotstuff}

\subsection{Integration Overview}

The $\DAL$ layer certifies that replica-contributed data is available and that
the leader has captured enough mini-blocks according to the DA-CR rule. To turn
this into censorship-resistant consensus, the surrounding BFT protocol must
enforce one additional condition: a leader should not be able to commit a BFT
block unless the block consumes a completed $\DAL$ output. We call this
enforcement rule \emph{\scooping}.

The integration is modular. We do not change the lock-commit safety argument of
the underlying leader-based BFT protocol. Instead, we strengthen the proposal
validity condition. A proposal is valid only if it is accompanied by a valid
$\DAL$ certificate for the corresponding view and commits to consuming the data
certified by that $\DAL$ instance. If the certificate is missing, invalid, or
unrelated to the proposal, honest replicas reject the proposal and do not
proceed to lock or commit it.

Scooping does not prescribe the application-level ordering or execution
semantics of the certified data. The BFT layer only checks that the proposal is
backed by a completed $\DAL$ instance. How the certified mini-blocks are later
derived into an application payload, ordered internally, executed, and served to
users is an upper-layer responsibility. For example, an application may use a
simple deterministic round-robin derivation over mini-blocks, or it may use a
more expressive ordering rule with additional validity evidence.

The rest of this section first defines the \scooping check, and then shows how
it is inserted into HotStuff-2~\cite{cryptoeprint:2023/397}. The key point is
that $\DAL$ is consumed as a deterministic validity predicate inside the BFT
protocol: HotStuff-2 still decides blocks through its normal voting and locking
rules, but honest replicas only vote for proposals that reference a completed
$\DAL$ instance and satisfy the scooping rule.

\begin{algorithm}[!t]
    {\fontsize{7.5pt}{7.5pt}\selectfont \caption{$\textit{Scooping}$}
    \label{alg:scooping}
    \begin{algorithmic}[1]
        \State \underline{\textbf{At the replicas }$p$: }
        \Procedure{$\textit{verifyCard}$}{ $commit, v$} \label{alg:scooping:line:validateCard_start}
            \State \Return \text{ms.verify(}\textit{aggSigs}\text{[} $commit,v$\text{])}
        \EndProcedure \label{alg:scooping:line:validateCard_end}
    \end{algorithmic}}
\end{algorithm} 

\subsection{{\scooping} Algorithm}
\label{sec:scooping_protocol}

Algorithm~\ref{alg:scooping} defines the minimal validity check inserted into
the BFT path. The procedure \textsc{verifyCard} checks local storage for a
valid aggregate signature produced by the $\DAL$ protocol for the corresponding
commitment and view. If the check succeeds, the replica has evidence that the
referenced $\DAL$ instance completed and that the data is available according
to the $\DAL$ certificate.

The check is deliberately small: mini-block collection, encoding, shard
verification, and aggregate-signature formation are handled entirely by the
$\DAL$ protocol in Section~\ref{sec:card_design}. The BFT layer only verifies
that the proposal refers to a completed $\DAL$ instance before the replica
votes, locks, or commits. A proposal that fails this check is treated as an
invalid proposal and rejected by honest replicas.

\subsection{{\nameOne} with HotStuff-2}
\label{sec:bigDipper_hotstuff}

We instantiate the integration with HotStuff-2~\cite{cryptoeprint:2023/397}.
The integration preserves the HotStuff-2 lock-commit structure: replicas still
use the same certificates, locks, and commit rule. The only change is that the
block-validity predicate is strengthened with the \scooping check. A replica
does not lock a HotStuff-2 block unless the block is backed by a completed
{\DAL} instance for the corresponding view.

The integration is lightweight because the four {\cardOne} messages align with
the four HotStuff-2 messages. HotStuff-2 uses \textsf{propose}, \textsf{vote},
\textsf{prepare}, and \textsf{vote2}; {\cardOne} uses \textsf{Dispersal},
\textsf{Approval}, \textsf{Agreement}, and \textsf{Collection}. We piggyback
each {\cardOne} message on the corresponding HotStuff-2 message:
\textsf{propose} carries \textsf{Dispersal}, \textsf{vote} carries
\textsf{Approval}, \textsf{prepare} carries \textsf{Agreement}, and
\textsf{vote2} carries the next-view \textsf{Collection}. Thus the integration
does not add a new BFT phase; it fills the existing HotStuff-2 message schedule
with the DA-CR messages needed for the next block.

Algorithm~\ref{alg:hotstuff-scooping} summarizes the modifications to
HotStuff-2, and Figure~\ref{fig:hotstuff-scooping} gives the full message flow
across views. The full integration details are given in
Appendix~\ref{sec:appendix_integration_hotstuff-2}. The important point for the
main text is that {\name} does not replace HotStuff-2's safety core; it only
prevents a leader from making progress with a proposal that has not scooped a
completed {\DAL} output. The integration fits all
DA-CR parts inside the two-phase BFT protocol, and requires no change to the pacemaker
during leader rotation. 

\begin{algorithm}[!t]
    {\fontsize{7.5pt}{7.5pt}\selectfont \caption{$\textbf{Hotstuff-2 adaptation}$}
    \label{alg:hotstuff-scooping}
    \begin{algorithmic}[1]
        \State \underline{\textbf{Modification to Enter(1) } }
        \State \textbf{As a leader}, no modification. if entering the view with double certificate, {$(C_{v-1}(C_{v-1}(B_{k-1})))$}, then proceed directly to Propose(2); otherwise the leader sets a timer of $3\Delta$ then proceeds to the Propose(2)
        \State \textbf{As a replica}, if entering a view with a double certificate, then proceed directly to Vote-and-commit(3); otherwise the replica sends its locked certificate and the corresponding valid combined signature of {\DA} to the leader along with \textbf{Collection} message prepared from Alg.{\cardOne} line~\ref{alg:card:line:signal_start}-~\ref{alg:card:line:signal_end}\label{alg:hotstuff-scooping:line:replica-enter}
        \State \underline{\textbf{Modification to Propose(2) } }
        \State The leader instead of include the BFT block, it performs Alg.{\cardOne} line ~\ref{alg:card:line:collection_start}-~\ref{alg:card:line:collection_end} to encode the collected mini-blocks, then broadcast the Propose and \textbf{Dispersal} message together. Any certified block must come with a valid combined signature, otherwise it 
        must not be used for extending the next block. \label{alg:hotstuff-scooping:line:leader_propose}
        \State \underline{\textbf{Modification to Vote and commit(3) } }
        \State Before proceeding any execution, the replica performs Alg.{\cardOne} line ~\ref{alg:card:line:dipserse_start}-~\ref{alg:card:line:dipserse_end} to checks the dispersal correctness and send \textbf{Approval} message, then proceeds to Hotstuff-2 \label{alg:hotstuff-scooping:line:vote_and_commit}
        \State \underline{\textbf{Modification to Prepare(4) } }
        \State The leader runs Alg.{\cardOne} line~\ref{alg:card:line:aggregate_start}-~\ref{alg:card:line:aggregate_end} and send \textbf{Agreement} message,
        then proceeds to Hotstuff-2
        \State \underline{\textbf{Modification to Vote2(5)} }
        \State Replicas performs procedures from Alg.Scooping ~\ref{alg:scooping:line:validateCard_start}-~\ref{alg:scooping:line:validateCard_end} before locking the BFT block. If succeed, send along with the vote2 message a \textbf{Collection} message prepared from Alg.{\cardOne} line~\ref{alg:card:line:signal_start}-~\ref{alg:card:line:signal_end}. \label{alg:hotstuff-scooping:line:optimistic_send_mini_block}
    \end{algorithmic}}
\end{algorithm}

\subsection{Pacemaker-Driven CARD Instances}
\label{sec:bigdipper:collection_waiting}

The abstract \event{Signal} used by {\cardOne} is supplied by the HotStuff-2
pacemaker. No pacemaker modification is required: when replicas enter view
\(v\), the same view number identifies the corresponding {\DAL} instance, and
replicas prepare mini-blocks for that view. Thus {\DAL} instances run at the
same frequency as BFT views. The leader starts dispersal after receiving enough
\textsf{Collection} messages and after the local condition
\textit{reach-proc-time} holds. This condition represents a tunable collection
deadline: waiting longer may include more honest mini-blocks, improving the
transaction-broadcast guarantees in Section~\ref{sec:navigation}, while waiting
less preserves lower latency. The detailed discussion of signal frequency and
waiting-time choices appears in Appendix~\ref{sec:card_design:freq_waiting}.

\subsection{Safety and Liveness}

The integration preserves HotStuff-2 safety because it only restricts the set
of proposals that honest replicas are willing to vote for. Honest replicas
continue to follow the ordinary HotStuff-2 voting, locking, and commit rules.
They additionally require the \scooping predicate to hold before locking or
committing a block. Since \scooping is a deterministic validity check for a
given proposal and view, it cannot create two conflicting certified blocks in
any execution where HotStuff-2 would prevent them. It can only cause honest
replicas to reject proposals that do not consume a completed $\DAL$ output.

The integration also preserves data availability for committed blocks. If an
honest replica locks or commits a $\name$ block for view $v$, then the block
passed \textsc{verifyCard} for the corresponding $\DAL$ commitment and view.
Therefore, the block is tied to a completed $\DAL$ instance. By the
\property{Availability}, \property{Correctness}, and \property{Commitment
Binding} properties of $\DAL$, the certified data is retrievable and uniquely
determined. By the DA-CR inclusion parameter of the selected $\DAL$
construction, the certified data also contains the required number of honest
mini-blocks.

For liveness, after GST, the HotStuff-2 pacemaker eventually selects an honest
leader and honest replicas enter the same view. The honest leader receives
enough collection messages, completes the $\DAL$ instance, and proposes a block
that passes \scooping. Honest replicas can then vote according to the ordinary
HotStuff-2 rules. Malicious leaders may fail to complete the $\DAL$ instance or
may propose blocks that fail \scooping, but those failures are handled by
timeout and view change. The full safety and liveness proofs are provided in
Appendix~\ref{sec:appendix_integration_proof}.

\section{Conclusion}
\label{sec:conclusion}
We present {\name}, a hyperscale BFT system that provides  short term censorship resistance to leader based BFT protocols. 
It is based on the idea that transaction inclusion can be made decentralized, and together with data availability protocols with censorship resistant properties, we can prevent a malicious leader from censoring the transactions. Client can use {\navigation} algorithm for controlling amount of censorship resistance with fast confirmation latency.



\bibliography{lipics-v2021-sample-article}

\section{Proof of Stake BFT Leader Assumption}
\label{sec:bft_leader_optin}
In the production leader based proof of stake blockchain, like Tendermint\cite{buchman2018latest}, a replicas' chance to become a leader is weighted proportional to replicas' voting power. 
There is a natural implication from this model to our assumption that leader is much capable such that it has high bandwidth and computing facility. 
Essentially, the leader with large percentage of stake in the partial synchronous setup already put down large capital investment in the form of stake, but the hardware requirement is relatively cheap. For example, a powerful instance m7g.16xlarge\cite{awsEC2} on AWS EC2 with $64$ vCPU, $256$ GB memory and $30$ $Gbps$ optimized bandwidth requires only $\$2.6112/hour$ in a long term plan. 

On the other hand, the censorship resistant property in {\name} does not give the leader much higher potential to manipulate the transactions inclusion. Therefore from the protocol's view, it is fine to have a few replicas with large amount of stake, and many replicas with less stake and regular hardware, provided that there is one powerful leader to drive liveness.

A more uniform stake distribution is desirable, in which scenario, replicas with insufficient hardware can choose third party services to temporarily serve the leader's role.
Or the protocol has to  downgrade the performance for a brief time proportionally to the replicas' stake.

To avoid the problem of downtime or third-party trust relation, the leader selection part of BFT protocol can have an opt-in mechanism, such that only interested parties can volunteer to becomes the leader. If some leaders are persistently censoring, eventually replicas from community can opt-in to improve the honest ratio in the leader set.

\section{BFT protocol complexity analysis}
\label{sec:appendix:bft_complexity_analysis}

Many existing BFT protocols have short term censorship resistance. In this section, we summarize those protocols and delineate their differences with {\name}.
We analyze each protocol by first identifying the Data availability part and consensus part. The comparison is shown in Table~\ref{table:bft_comparison_detail}.
In many asynchronous protocols including Tusk\cite{danezis2022narwhal}, VABA\cite{abraham2018validated},Bullshark\cite{spiegelman2022bullshark},
Dag-Rider\cite{keidar2021all}, there is a concept of a passive leader, who is retrospectively selected after a round has passed, the purpose of such leader is to create an anchor to determine of head of a set of confirmed block. The leader does not serve the role to prevent the all-to-all communication. It is different from the active leader from the classic leader based BFT protocol, in which a leader actively coordinates among BFT replicas and make sure the leader is progressing.

\begin{table}[!hbt]
\centering
\caption{Comparison of BFT protocol in a system of $n$ replicas}
\begin{subtable}[h]{\linewidth}
\centering
\resizebox{1\textwidth}{!}{
\begin{tabular}{|c | c | c | c | c | c |c | c| c | c | c |}
\hline
\multirow{3}{*}{}  BFT  & Coor-  & \multicolumn{3}{c|}{\textbf{Per Replica Communication}}   & System & Number & Block & Worst Case & Max \\
\cline{3-5}
        protocols & dinating & \multicolumn{2}{c|}{\textbf{DA }} & \multicolumn{1}{c|}{\textbf{Consensus}} & / Replica & Comm & Interval & Inclusion Distance & Adversary\\    
\cline{3-4}        
        & Leader & Complexity & Download & Complexity & Throughput & Steps & (seconds) & (blocks) & Ratio \\
\hline \hline
Hotstuff\cite{yin2019hotstuff} & \greencheck & $O(nb)$ & \greencheck & $O(1)$ $\dagger$ & $O(1)$ & $4$ & $O(1)$ & $f+1$ & $0.33$ \\
\hline
HoneyBadger\cite{miller2016honey} & \redcross & $O(nb + \lambda n^2\log{n})$ & \greencheck & $O(n^2\log{n})$  & $O(1)$ & $3+\log{n}$ & $O(n^2\log^2{n})$ & $1$ & $0.33$ \\
\hline
DispersedLedger\cite{yang2022dispersedledger} & \redcross & $O(b+\lambda n^2)$ & \redcross & $O( n^2\log{n})$ & $O(n)$ & $3+\log{n}$ & $O(n^2\log^2{n})$ & $1$ & $0.33$\\
\hline
Prime\cite{amir2010prime} & \greencheck & $O(nb + n^2)$& \greencheck & $O(n^2)$ & $O(1)$ & $6$ & $O(n^2)$ & $1$ & $0.33$\\
\hline
Dumbo-2\cite{guo2020dumbo} & \redcross & $O(nb + \lambda n^2\log{n})$ & \greencheck & $O(\lambda n^2)$ & $O(1)$ & $20$ & $O(n^2\log{n})$ & $1$ & $0.33$ \\
\hline
Dumbo-NG\cite{gao2022dumbo} & \redcross & $O(nb)$ & \greencheck & $O(\lambda n^2 \log{n})$ & $O(1)$ & $9$ & $O(n^2\log{n})$ & $1$ & $0.33$\\
\hline
VABA\cite{abraham2018validated} & \redcross & $O(nb)$ & \greencheck & $O(n)$ & $O(1)$ & $13$ & $O(n)$ & $1$ & $0.33$ \\
\hline
Dag-Rider+AVID\cite{keidar2021all} & \redcross & $O(nb + \lambda n^2\log{n})$ & \greencheck & $O(n^2\log{n})$ & $O(1)$ & $12$ & $O(n^2\log{n})$ & $1$ & $0.33$\\
\hline
Narwhal-Hotstuff\cite{danezis2022narwhal} & \greencheck & $O(nb + \lambda n^2)$ & \greencheck & $O(1)$ & $O(1)$ & $6$ & $O(n^2\log{n})$ & $1$ & $0.33$\\
\hline
Tusk\cite{danezis2022narwhal} &  \redcross & $O(nb + \lambda n^2)$ & \greencheck & $O(n^2\log{n})$ & $O(1)$ & $9$ & $O(n^2\log{n})$ & $1$ & $0.33$\\
\hline
Bullshark\cite{spiegelman2022bullshark} & \redcross & $O(nb)$ & \greencheck & $O(n^2\log{n})$ & $O(1)$ & $4$ & $O(n^2\log{n})$ & $1$ & $0.33$ \\
\hline
Bullshark-Parital\cite{spiegelman2022bullshark} & \redcross & $O(nb)$ & \greencheck & $O(n^2\log{n})$ & $O(1)$ & $4$ & $O(n^2\log{n})$ & $1$ & $0.33$ \\
\hline
$\nameQuarter$ & \greencheck & $O(b)$ & \redcross & $O(1)$ & $O(n)$ & $4$ & $O(1)$ & $1$ & $0.25$\\
\hline
{\nameOne} & \greencheck & $O(b + \lambda n)$ & \redcross & $O(\lambda n)$ & $O(n)$ & $4$ & $O(n)$ & $1$ & $0.33$ \\
\hline
$\nameLite$ & \greencheck & $O(b + \kappa \lambda \log^2{n})$ & \redcross & $O(1)$ & $O(n)$ & $4$ & $O(\log^2{n})$ & $1$ & $0.33$\\
\hline
\end{tabular}
}
\caption*{$\dagger$: $O(n)$ worst comm complexity\cite{cryptoeprint:2023/397}, when all subsequent leaders are malicious.}
\end{subtable}
\label{table:bft_comparison_detail}
\end{table}

In the table, we decompose communication complexity per replica into Data availability and consensus. 
Data availability of system does not require downloading the whole data at each replica. To signify the difference, we add a sub-column, \textbf{Download}, to indicate if the DA complexity already take care of the whole data availability in each replica as a part of DA protocol. 
We note that by decouple the retrieval with dispersal, we are able to remove a order $O(n)$ factor in front of the data workload $b$. 
The number of communication step depicts how many point-to-point (i.e. half of a round trip) in order for a transaction to get confirmed when the system is operating under an optimistic path.
It is a complementary measure to the latency in additional to complexity analysis.

For the rest of the section, we summarizes each protocol and points out their difference with {\name} protocols.

\subsection{Prime}
\label{sec:appendix:bft_complexity_analysis:prime}
Prime\cite{amir2010prime} is a classic leader based protocol whose goal is to defend against adversarial attack on performance degradation while ensuring basic safety and liveness of a BFT protocol. 
The malicious replicas can perform arbitrary operations but correctly enough for not being identified as malicious.
The paper defines a new correctness criteria that takes performance into account for defending against performance degradation attack.
For example, a malicious leader can introduce latency in the consensus critical path to delay the rate of confirmation.
The paper relies two insights, most of procedures in the protocol does not require malicious replicas to take action; second, if the workload for a leader has a predictable bound, all peers can observe the network and access the leader's performance to decide if the leader should be replaced.
To achieve them, the protocol devises multiple sub-protocols, each responsible for a specific tasks. When a replica receives a order (transaction), it disseminate it to all other replicas. The pre-ordering sub-protocol is responsible for confirming the data availability. A data structure called {\textbf{PO-SUMMARY}} is invented to record from each replica's view, how much data the other replicas currently have. Essentially, {\textbf{PO-SUMMARY}} is a vector recording cumulative sum about what other replicas have. Then each replica broadcast {\textbf{PO-SUMMARY}} to everyone including the leader.  The exchange about {\textbf{PO-SUMMARY}} requires a $O(n^2)$ communication, where $n$ is the number of replicas. 
The leader then uses all {\textbf{PO-SUMMARY}} to create a global ordering, if a transaction has been recorded by at least $2f+1$ replicas. 
{\name} is different from Prime although both are leader based protocols, because in {\name} there is never a all-to-all communication pattern, and in addition, the inclusion of transaction does not require consensus of $2f+1$ to receive the transaction. 
%
%
Although Prime uses erasure code to reconciliate data availability, the erasure coding is done to per transaction basis. The regular Ethereum transaction size if roughly 200Byte, then if the network has size of equal size, the overhead per transaction is outweighting the benefits of erasure coding.
Whereas in {\name}, the erasure code is done all transactions batch received by a BFT replica. Applying erasure code to the batch does not introduce too much overhead for indicating which indices its points are evaluating on.
At latency perspective, Prime extends 2-Phase PBFT by 1.5 phase (round-trip) trips message, which are {\textbf{PO-REQUEST}, \textbf{PO-ACK}, \textbf{PO-SUMMARY}}, whereas {\name} perfectly fits into the 2-Phase hotstuff protocol.

\subsection{HoneyBadger and DispersedLedger}
\label{sec:appendix:bft_complexity_analysis:honeybadger_and_dispersed}

HoneyBadger is a leaderless asynchronous BFT protocol, made of two protocols: a reliable broadcast protocol (RBC) and an asynchronous binary agreement protocol. 
To reach consensus, HoneyBadger runs $n$ reliable broadcast protocol asynchronously. The reliable broadcast is modified based on AVID~\cite{cachin2005asynchronous}, where every replica produces its own merkle tree, and broadcast the merkle leaf and proof to all other replicas. 
The total communication complexity for one replica to broadcast its data of size $b$ bytes is $O(n b + \lambda n^2 \log{n} )$; whereas using the classic Bracha\cite{bracha1987asynchronous} reliable broadcast would take $O(n^2 b)$.
Whenever a replica completes a reliable broadcast $j$, it starts to contribute to the binary agreement protocol for $j$. 
An asynchronous binary agreement protocol accepts inputs from all replicas. 
Since the binary agreement(BA) protocols based on construction from \cite{mostefaoui2014signature} requires $O(n^2)$ communication complexity per round, and it takes $O(k)$ rounds to complete with probability $1 - 2^{-k}$. The overall complexity for BA is $O(n^2\log{k})$.
To ensure all $n$ instances of BA to terminate with high probability, the expected number of round is $\log{n}$.
The overall communication complexity for all $n$ replicas to complete a round is $O(n^2 b + \lambda n^3 \log{n} + n^3\log{n}) = O(n^2 b + \lambda n^3 \log{n})$, where $b$ is inputs from each replicas.
The number of communication steps to complete a honeybadger RBC (which is modified based on AVID) is $3$. The first step sends merkle leaf to everyone, the second step sends echo among all replicas, and the final step for sending ready messages.

DispersedLedger has a very similar structure as HoneyBadger. DispersedLedger divides the protocols into two parts: a data availability(DA) protocol and a binary agreement protocol. 
The data availability is different from reliable broadcast protocol in HoneyBadger, because data availability can be achieved by the \textsf{Disperse} invocation alone, and its retrieval part can be deferred in the later stage. 
%
%
The construction of DA in DispersedLedger is called AVID-M, which has very similar structure as the merkle construction in reliable broadcast in HoneyBadger and AVID-H in AVID\cite{cachin2005asynchronous}.
In DispersedLedger, the communication complexity for each replica is $O(b + \lambda n^2)$, the binary agreement part is identical to HoneyBadger which takes $O(n^2\log{k})$ per replicas, that terminates with probability $1-2^{-k}$. So together the total communication complexity is $O(bn + \lambda n^2 + n^3\log{n})$.
The number of steps for DA is $3$, the sender replica first sends the merkle leaf and proof, the each replica sends {\textbf{GotChunk}} that contains only the merkle root, this is the step that differs significantly from AVID-modified RBC from HoneyBadger; the last step is to send {\textbf{ready}} message.
Although both DispersedLedger and HoneyBadger has censorship resistance property, it is implemeneted by a all-to-all communication pattern. Whereas in {\name}, there is an active coordinating leader preventing the quadratic communication pattern.

\subsection{VABA}
\label{sec:appendix:bft_complexity_analysis:vaba}

Validated Asynchronous byzantine Agreement is the first protocol that achieves multivalue byzantine agreement with $O(n^2)$ complexity. Previous method by \cite{cachin2001secure} of MVBA requires $O(n^3)$ communication complexity.
The key construction element is called provable broadcast, which is also a key building block for conventional leader-based BFT protocol like~\cite{yin2019hotstuff,buchman2018latest,castro1999practical}, where the leader has an active role to coordinate replicas. For VABA, each Atomic broadcast instance is associated with an ID; within each ID, four consecutive PB are invoked to ensure data is committed without worries of hidden lock. All $n$ replicas performs 4-stages PB, and a leader is chosen retrospectively after all replicas have abandoned their PB procedures.
If the selected leader prematurely terminates the 4 stage PB, a view change repeats the 4 stage PB until an externally valid value is confirmed by all honest replicas.
{\name} is a leader based BFT protocol, that avoids all-to-all communication pattern, unlike VABA.

\subsection{Dumbo, sDumbo-DL and Dumbo-NG}
\label{sec:appendix:bft_complexity_analysis:dumbo}

Dumbo is BFT protocol built upon the framework presented from HoneyBadger. There are two versions of the protocol: Dumbo1, and Dumbo2. The Dumbo1 identified $n$ parallel asynchronous Binary Agreement protocol as the bottleneck to the system throughput. Then it provides an optimization to replace $n$ Asynchronous Byzatine Agreement(ABA) instance. The technique is to replace $n$ ABA with a $\log{n}$ reliable broadcast(RBC) to ensure one honest replica can provide a batch of transactions from $n-f$ replicas. Then all $n$ nodes runs {$\log{n}$} ABA to agree on which one of the $\log{n}$ Index from RBC are correct, and the data is received by them.
In total the step which Dumbo1 takes to finish the protocol is $2 + 2 + \log{\log{n}}$ for Value-RBC, Index-RBC and ABA.

Dumbo2 is a further optimized version of Dumbo1. Instead of using $n$ Value-RBC, ${\log{n}}$ Index-RBC followed by ${\log{n}}$ ABA, Dumbo2 replaced the reliable broadcast part with a new protocol called provable RBC, that uses threshold signature to ensure the data broadcast by Index-RBC are indeed correct. 
The agreement part is replaced with only one Multivalue Byzantine Agreement (MVBA).

\subsection{DAG-Rider, Narwhal-Hotstuff and Tusk}
\label{sec:appendix:bft_complexity_analysis:dag-rider}

DAG-Rider, Narwhal-Hotstuff and Tusk are consensus protocols that adopts a design that separate mempool and consensus, where the mempool is implemented as Direct acyclic graph.
DAG-rider is an asynchronous protocol that arrives at consensus every four rounds, where each round consists of a reliable data broadcast (RBC) and is not a point-to-point communication. 
When a RBC is implemented by AVID, the communication complexity is $O(nb + \lambda n^2 \log{n})$.
Because each RBC takes $3$ communication steps (point to point communication), the overall steps it takes to confirm a block is $12$. Because every vertex must contain $n$ meta-data of size $\log{n}$, the broadcast would consume $O(n^2\log{n})$ bits for each replica.

Narwhal is a mempool module that is dedicated to produce data protocol that has properties of Integrity, block-availability, containment, 2/3-causality and 1/2-chain quality.
The data broadcast is achieved with provable broadcast by using certificate from $n-f$ replicas' signature.
Narwhal-Hotstuff-2 is a partial synchronous protocol that uses Hotstuff-2 for reaching consensus. In the table, we use the latest Hotstuff-2 as the consensus module to reduce latency. In the Narwhal-Tusk paper, the communication steps is $8$, because it used the old $3$ phases Hotstuff protocol.

Tusk is an asynchronous protocol that is an implementable version of DAG-Rider, and it includes optimization to reduce the expected communication steps to $9$ as indicated in the paper\cite{danezis2022narwhal}.

\section{Background Appendix }
\label{sec:appendix_background}

\subsection{Validity requirement in BFT problem}
\label{sec:appendix:background:validity}
Often external validity\cite{cachin2001secure} is added to the BFT requirement, which requires a polynomial-time computable predict for deciding if a transaction is valid.
{\name} uses standard Virtual machine to handle transaction validity, the protocol by itself serves to reach consensus on a sequence of bits. 

\subsection{DA using outsourced components}
\label{sec:appendix:background:da}
There are many ways to incorporate \cite{sheng2020aced, al2019lazyledger} a DA into a BFT protocol, either by borrowing external DA oracle or internalizing it with the BFT replicas. The advantage of the second option is a single unified trust assumption. The advantage of former options usually come with high throughput and simpler application design logic due to modularity.

\subsection{Reed Solomon Erasure coding}
\label{sec:appendix:background:rs}
The encoding process of Reed-Solomon code is based on the idea of polynomial interpolation\cite{reed1960polynomial}.
Essentially, given a list of $n$ points, there is a unique polynomial of degree $n-1$ that interpolates those points.
The redundancy generation is basically evaluating the polynomial at more points. Where those new points locate depends on specific encoding scheme. 
However, one encoding scheme that uses Fast Fourier Transform is particularly popular due to its fast $O(n\log{n})$ computation complexity. It only requires the field which evaluations indices belong to have root of unity, $\omega$, such that $\omega^z=1$, where $z$ is power of 2.
Given some data $D$, the exact operation can be summarized as $IFFT(\textit{concate}(FFT(D), \widehat{0}))$, where $IFFT$ is inverse FFT, and the length of $\widehat{0}$ depends on size of redundency.
Since Fast Fourier Transform is a linear operation, i.e. $\alpha FFT(a) + \beta FFT(b) = FFT(\alpha a + \beta b)$.
We will use the FFT encoding scheme on Elliptic Curve\cite{zcashHalo2} to work with KZG polynomial commitment.%
Given an input array $[d_1 \cdots d_n]$ with coding ratio $1/3$,
RS encoding can be implemented such that the output is $[d_1 \cdots d_n, p_{n+1} \cdots p_{3n}]$, where the input data take identical locations, and parity data $p_{*}$ are generated at new indices.

RS encoding can be used to introduce redundancy in numerous ways. Figure~\ref{fig:kzg_1d} presents a way by laying all data in one dimensional line, and apply RS encoding.
Another way is to format the data into a 2D matrix like 
Figure~\ref{fig:card_data_struct}.
Inside the matrix, data are partitioned into columns, and redundancy is applied by taking FFT extension on each rows, such that redundant columns are generated.
The matrix format is pioneered by~\cite{al2019lazyledger} and used by Ethereum Danksharding\cite{dankshardinghackmd} in the future roadmap. Unlike those works, our construction the matrix does not add redundancy along columns.
Our RS construction is also different from Semi-AVID-FR\cite{nazirkhanova2021information}, that every replica would possess one system chunks and two parity chunks.

\subsection{BLS Elliptic Curve and Signature}
\label{sec:appendix:bilinear-curve}
This section serves as a very brief coverage of basic elements of BLS12-381 elliptic curve and their properties. We also briefly cover BLS signature which can be instantiated on the BLS12-381 curve.

An elliptic curve consists of three group $G_1, G_2, G_T$ or order $p$, a pairing is an efficiently computable function $e: G_1 \times G_2 \rightarrow G_T$. Group $G_1$, $G_2$ have generator $g_1, g_2$ respectively, and all members in the group can be obtained by enumerating powers on the generator. The generator of $G_T$ equals to $e(g_1, g_2)$.
A pairing function provides a following property that $e(X^a, Y^b) = e(X,Y)^{ab}$, where $X,Y$ are elements from elliptic curve group $G_1, G_2$, and $a,b$ are integer from their corresponding finite field.
BLS12-381 is one elliptic curve allows for pairing (bilinear) operations.

On top of BLS12-381 curve, we can instantiate BLS public key signature~\cite{boneh2001short} scheme, which allows clients to sign and verify signature.
Beyond simple signature, both threshold signature and multi-signature can be developed.

\subsection{Polynomial Commitment and KZG}
\label{sec:appendix:background:kzg}
We provide a more technical definition for primitives that uses KZG scheme that uses Discrete Logarithm(DL) and t-SDH assumptions~\cite{kate2010constant}. 
First, KZG requires a trusted setup of $\{g^z \cdots g^{z^t}\}$, where $z$ has to be a secret to anyone and $g$ is a generator for some elliptic curve group that has bilinear pairing property~\cite{boneh2018compact}. We use \textsf{Setup} to presents the procedures to get those powers of generators, which requires a security parameter $\lambda$.
The security of a scheme is measured against a probabilistic polynomial time (PPT) adversary to break the properties. A secure scheme requires all polynomial time adversary to have negligible probability to break the scheme.
The primitive $Commit$ of a function $\phi = \sum_j \phi_j x^j$ can be implemented as $C = \prod_{j=0}^{deg(\phi)} {(g^{z^j})}^{\phi_j}$, where $g^{z^j}$ is provided from the public setup, and $\phi_j$ are the coefficient for $j$-th monomial.
The primitive $creatWitness$ for some index $i$ with evaluation $\phi(i)$ is defined as $\psi(x) = \frac{\phi(x) - \phi(i)}{x-i}$, and we let $w_i = g^{\psi(z)}$, which is the outcomes of $Commit(\psi)$.
The primitive $VerifyEval$ on some evaluation a tuple $(i, \phi(i), w_i)$ is defined as $e(C, g) = e(w_i, g^z / g^i)e(g, g)^{\phi(i)}$, where $e$ represents the bilinear pairing operation.
A proof for its correctness can be found in~\cite{kate2010constant}.

\begin{definition}(KZG Commitment scheme)
   \begin{enumerate}
       \item \textbf{Polynomial-binding}. For all PPT adversaries $\mathcal{A}$ 
       \begin{align*}
        \operatorname{Pr}
        \left(
            \begin{array}{c}
                pk \leftarrow Setup(1^{\lambda})\\
                (\kzgcommitment, \phi(x), \phi'(x)) \leftarrow \mathcal{A}(pk)
            \end{array}
            :
            \begin{array}{c}
                Commit(\phi) = C \wedge \\
                Commit(\phi') = C \wedge \\
                \phi \neq \phi'
            \end{array}
        \right) \leq \epsilon(\lambda).
        \end{align*}
        \item \textbf{Evaluation-binding}. For all PPT adversaries $\mathcal{A}$ 
       \begin{align*}
        \operatorname{Pr}
        \left(
            \begin{array}{c}
                pk \leftarrow Setup(1^{\lambda}), \\
                (\kzgcommitment, \langle i,\phi'(i)), w_i\rangle, \\
                \langle i,\phi'(i)), w_i' \rangle \\
                \leftarrow \mathcal{A}(pk)\\
            \end{array}
            :
            \begin{array}{c}
                VerifyEval(pk, C, i, \\
                    \phi(i), w_i) = 1 \wedge \\
                VerifyEval(pk, C, i, \\
                    \phi'(i), w_i') = 1 \wedge \\
                \phi \neq \phi'
            \end{array}
        \right) \leq \epsilon(\lambda).
        \end{align*}
   \end{enumerate} 
\end{definition}

$\epsilon$ represents a negligible function such that $\epsilon(n) \le \frac{1}{n^c}$ for all n greater some function of $c$. $\mathcal{A}$ denotes an adversary.

It can be shown that KZG provides linearity in commitments and witnesses, such that $commit(\alpha u + \beta v) = \alpha commit(u) + \beta commit(v)$, $\textit{createWitness}(\alpha u + \beta v)=\alpha \textit{createWitness}(u) + \beta \textit{createWitness}(v)$. Here we slightly abuse the notation a bit to use $+$ to present the group operation $\cdot$, it is useful to illustrate the following property.

The linearity of KZG commitment and RS encoding provide a powerful way to encode systematic chunks such that commitments of parity chunks interpolates to a polynomial, which is generated purely from commitments of systematic chunks. 

There are existing works ~\cite{hendricks2007verifying,nazirkhanova2021information,dankshardinghackmd} that study and apply this property.
Generating witnesses(proofs) is a computational intensive process. Techniques including KZG multi-reveal~\cite{kateamortized} and Universal Verification equations~\cite{universalverify} from Ethereum provide both fast way to compute batch and evaluate proofs.

\subsection{Pedersen commitment}
\label{sec:appendix:background:pedersen_commitment}
Pedersen commitments can be used to commit to a vector of field elements,
and hence to the coefficient vector of a polynomial. Let \(G\) be a cyclic
group of prime order \(q\), written multiplicatively. Let
\(\mathbf g=(g_1,\ldots,g_n)\) be group elements generated so that no party
knows a nontrivial discrete-log relation among them. For a vector
\(\mathbf a=(a_1,\ldots,a_n)\in \mathbb F_q^n\), define
\[
    \operatorname{commit}(\mathbf g,\mathbf a)
    =
    \prod_{i=1}^n g_i^{a_i}.
\]
This commitment is additively homomorphic:
\[
    \operatorname{commit}(\mathbf g,\mathbf a)\cdot
    \operatorname{commit}(\mathbf g,\mathbf b)
    =
    \operatorname{commit}(\mathbf g,\mathbf a+\mathbf b).
\]
Under the discrete-log assumption, and assuming the bases
\(g_1,\ldots,g_n\) are generated without known relations, the commitment is
computationally binding: it is infeasible to find two distinct vectors
\(\mathbf a\neq \mathbf a'\) such that
\[
    \operatorname{commit}(\mathbf g,\mathbf a)
    =
    \operatorname{commit}(\mathbf g,\mathbf a').
\]

\subsection{Ordering Functions}
\label{sec:appendix:background:oracle}  
{\name} prioritizes transaction inclusion, and provide a modular interface for any ordering modules. 
From the perspective of {\name}, the other end of the interface is treated as a \textbf{ordering} function, that returns a deterministic permutation of transactions for data corresponds to some BFT view number.
For each replica, the function could be a simple deterministic implementation. For example, a shuffling function can be implemented by computing the hash digest of all sorted transactions; another example would be using round robin to take transaction one-by-one from each replica.
Then if all replica only access to the same data, they all arrive to the same permutation.
However, simple ordering function suffers a limitation that a malicious leader can grind its mini-blocks to alter transaction locations.
More sophisticated oracle usually requires inputs from other replicas or external source for providing the ordering. 
For example, logical timestamp among transactions is required for creating a FCFS fair ordering\cite{kelkar2021themis}; 
auction mechanism needs to collect bids from possible external parties;
a verifiable random function(VRF) requires randomness inputs from other replicas. 
Most implementations can be optimized with coordination from the leader.
%
%
%


\section{{\DA} relationship with AVID protocol family}
\label{sec:appendix:card:avid}
AVID is not communication efficient since every replica downloads the entire data. AVID-FP fixes it by using homomorphic fingerprinting, but requires a $O(n^2 \lambda)$ number of overhead per node.
AVID-M overcomes the quadratic overhead by using merkle commitment, but it allows the room for inconsistent coding, which either complicates the trust assumption or prolongs the confirmation latency.
Semi-AVID-PR is designed for rollup, so it does not provide \textbf{Agreement} and \textbf{Termination} properties unlike others, but it shares many similarity with {\DAL} on data structure and properties.
Compared to those protocols, there are two great differences in {\DAL}.
First, {\DAL} has a network assumption of partial synchrony, which allows for a leader to intermediate among replicas; it is the key that allows DA-CR to avoid the second $O(n)$ broadcast step in AVID-M\cite{yang2022dispersedledger} needed to ensure the Termination property.
Second, AVID family protocols has a client-server model and the client has the entire data; whereas {\DAL} is closer to a P2P network, where every replica has partial data and serves as both a client and a server.
The {\DAL} is similar to AVID-FP, but avoids $O(n^3)$ messages complexity, 
because the leader can coordinate communication in the partial synchronous network.
%
%
At the downside of the {\DAL}, a malicious leader can refuse to proceed the protocol at any stage. 
A timeout and a leader replacement mechanism are required to properly terminate a round. 



\section{{\DA} properties discussion}
\label{sec:appendix:card_property:discussion}

We note that the \property{Inclusion}-$t$, \property{Data-tampering Resistance} and \property{Space-Capturing}-$\eta$ are extensions to the \property{Correctness} property which is offered by many {\DA} protocols.
Those additional properties refine the relation between the original data from an honest replica and its final dispersed mini-block even when the leader is malicious. 
\property{Data-tampering Resistance} is different from \property{Commitment Binding}, because the later guarantees there is only one data block that matches to the final commitment and its corresponding combined signature;
whereas for \property{Data-tampering Resistance} in an accountable mechanism, all dispersed mini-blocks must indeed be consistent to the published data attestations binding to the final commitment. 

Otherwise the leader is free to mutate any bits in any replicas' mini-blocks. The malicious leader can then attack the system by posting the original attestations, but in fact dispersing something else. 
Without careful protocol design, one can come up with scheme that violates the data tampering.
The \property{inclusion}-$t$ property directly corresponds to the censorship resistant in Definition~\ref{def:censorship-resistance},
and similarly, {\property{Space-Capturing}-$\eta$} for Definition~\ref{def:space-capturing}.

\section{Calculating Coding Ratio for $\DAL$}
\label{sec:appendix_coding_ratio}
The coding ratio is set to be $1/3$, i.e. 2 parity chunks per one systematic chunk. The number three comes from the assumption that only $f$ out of $3f+1$ BFT nodes are malicious.
To prevent liveness attack from $f$ malicious replicas, the leader should be allowed to proceed to complete a $\DAL$ instance as long as there is a combined signatures from $2f+1$ replicas. But since $f$ out of the $2f+1$ signatures could come from the malicious replicas, we require the remaining $f+1$ replicas to be able to reconstruct the whole data. 
Although the exact ratio should have been $\frac{f+1}{3f+1}$, the erasure coding is more efficient for integer and easy for reference. Since $\frac{1}{3} < \frac{f+1}{3f+1}$, the data is safe. 

\section{Retrieve Protocol}
\label{sec:appendix:retrieve_protocol}  
\subsection{Retrieve Algorithm}
\label{sec:appendix:retrieve_algorithm}  
A retriever can invoke either {\proceduref{retrieveChunk}} or {\proceduref{retrievePoint}} in Algorithm~\ref{alg:retrieve} for receiving the data.
\begin{algorithm}[!htbp]
    {\fontsize{7pt}{7pt}\selectfont \caption{Retrieve Invocations}
    \label{alg:retrieve}
    \begin{algorithmic}[1]
        \State \underline{\textbf{At the replica } $p$ :}
        \Event{$ \textit{receiving } \textsf{RetrieveChunk}$}
            \If{$ m.c_p \text{ in } db$}              
                \State $\text{send back } \langle db[m.c_p] \rangle $
            \EndIf
        \EndEvent
        \Event{$ \textit{receiving } \textsf{RetreivePoint}$}
            \If{$ m.c_p \text{ in } db$}    
                \State $\phi_{j}, w_j \gets \textit{createWitness}(db[m.c_p], j)$
                \State $\text{send back } \langle \phi_{j}, w_j \rangle $
            \EndIf
        \EndEvent
        \State         
        \Procedure{$\text{retrieveChunk}$}{ $p$, $\{c_i\}_{1 \le i \le n}\}$ } $\textit{\# if p=-1, then retrieve all chunks}$
            \If{$p \neq -1$}
                \For{$k \gets \{p, l\}$}
                    \State $chunk \gets send \langle RetrieveChunk, c_p \rangle \text{ to replica } k $
                    \If{$checkCom(c_p, chunk)$} 
                        \State \Return $chunk$
                    \EndIf
                \EndFor
            \EndIf
            \State $chunks = \emptyset$ $\quad \textit{\# when not comply}$
            \State $\{c_i\}_{0 \le i \le 3n} \gets extendEval(\{c_i\}_{1 \le i \le n})$
            \For{$j \gets \text{shuffle}(1 \cdots n) $} 
                \State $chunk \gets send \langle \textsf{RetrieveChunk}, c_j\rangle \text{ to replica } j$
                \If{$checkCom(c_j, chunk) = \true$}
                    \State $chunks[j] \gets chunk$
                \EndIf
            \EndFor
            \Event{$len(chunks) = f+1$ }
                \State $B \gets rsReconstruct(chunks)$
                \If{$checkCom(c_p, B[\cdot, p])$} $\textit{\# if p=-1, loop for all chunks}$
                    \State \Return $B[\cdot, p]$
                \EndIf
            \EndEvent
        \EndProcedure
        \Procedure{$\text{retrievePoint}$}{ $p$, $j$, $\{c_i\}_{1 \le i \le n}\}$}  $\textit{\# for a point at B[j,p]}$
            \For{$k \gets \{p, l\}$}
                \State $\phi, w \gets send \langle RetrievePoint, c_p, j\rangle \text{ to replica } k$
                \If{$\textit{kzgEvalVerify}(c_p, j, \phi, w)$} 
                    \State \Return $\phi$
                \EndIf
            \EndFor
            \State $\{c_i\}_{0 \le i \le 3n} \gets extendEval(\{c_i\}_{1 \le i \le n})$
            \For{$i \gets 1 \cdots n$}
                \State $chunk \gets send \langle RetrieveChunk, c_j\rangle \text{ to replica } j$
            \EndFor
            \State $B \gets \textit{retrieveChunk}(-1, \{c_i\}_{0 \le i < 3n}\})$ $\quad \textit{\# when not comply}$
            \State \Return $B[j,p]$
        \EndProcedure
    \end{algorithmic}}
\end{algorithm}

\subsection{Retrieval success as a function of number requested replicas}
\label{sec:appendix:retrieve:num_reconstruct}
Because every replica has three data chunks, a retriever at most needs to query $\left \lceil \frac{f+1}{3} \right \rceil + f$ replicas to reconstruct the interested rows in the data matrix. Let $a = \left \lceil \frac{f+1}{3} \right \rceil$.
Figure~\ref{fig:num_conn_to_success} plots the probability of successful reconstruction as a function of random connections with varying number of nodes $n=3f+1$.
Since each honest replica randomly select others to connect with, the probability of connecting to $a$ honest replicas, $\Pr(success)$, can be calculated by the following Equation: 

\[
    \label{eq:2}
    \Pr(success)  =\sum_{b=a}^y \frac{{2f+1 \choose b}{f \choose y-b}} {{3f+1 \choose y}} \tag{1}
\]

As we increase the number of replica, the ratio of achieving successful reconstruction is high even at the ratio of 0.6(f+a). Each node only need to connect to approximately $4f/5$ replicas to reconstruct the data if it uses query replicas randomly.

The data reconstruction can be run asynchronously, and takes one round trip time to complete. 

\begin{figure}[h]
\centering
\includegraphics[width=0.5\textwidth]{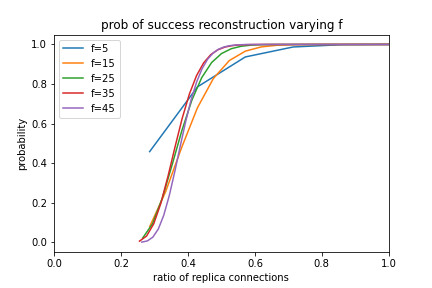}
\caption{The probability to connect to $a$ replicas as we change the total number of replicas $n=3f+1$. The x axis is presented as the ratio of $\frac{y}{f+a}$. The max value of x-axis is $f+a$.} 
\label{fig:num_conn_to_success}
\end{figure}

\begin{figure}[h]
\centering
\includegraphics[width=0.5\textwidth]{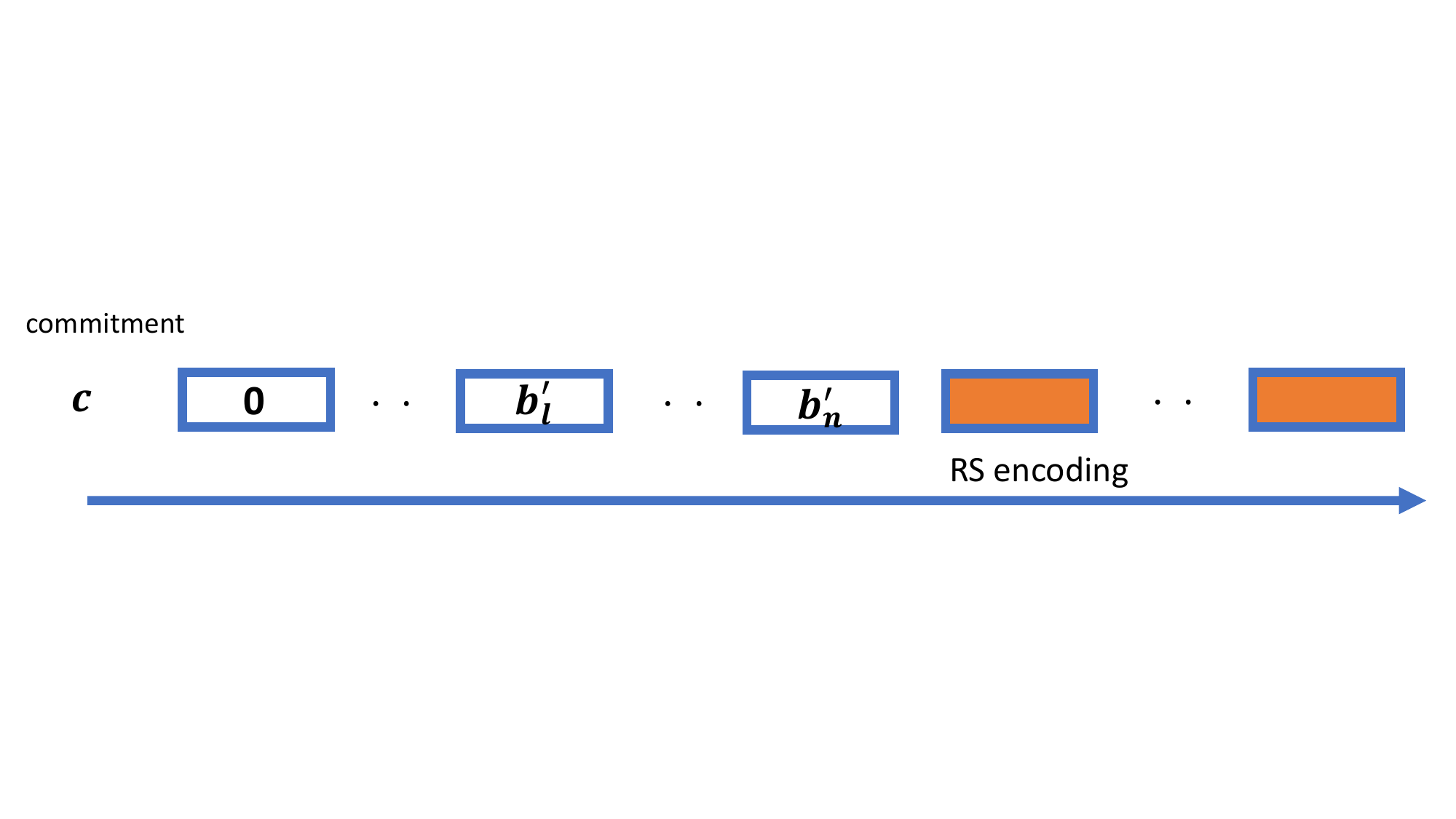}
\caption{The data structure for 1 Dimensional KZG RS encoding}
\label{fig:kzg_1d}
\end{figure}

\section{{\cardOne} Property proofs}
\label{sec:appendix:card7_proof}
\textbf{Termination}: Suppose some replicas invoke \primitive{Disperse}$(B)$. Because all honest replicas can timeout, they eventually finish with type \event{Incomplete} even if the protocol has stalled due to malicious leader or network asynchrony. 
%
When the leader is honest and the network is after GST, 
there would be at least $2f+1$ \messageProto{Collection messages} with correct commitment signatures. All verification in the \messageProto{Dispersal messages} would pass, and the leader is able to gather at least $2f+1$ aggregate signatures. As the result, all honest replicas can finish with type \event{Complete}.

\textbf{Availability}: An honest replica has finished with a $\DAL$ instance with a type \event{Complete}, it must have received an aggregate signature with at least $2f+1$ signer on a common message. $f+1$ of them must come from honest replicas. Because the coding ratio is $\frac{1}{3} < \frac{f+1}{3f+1}$, there are sufficient honest replicas to complete reconstruction. As the result any honest retriever can retrieve and reconstruct some data $B'$.

\textbf{Commitment Binding}: We prove by contradiction. Suppose a dispersal of $B$ completed with an aggregate signature, and adversary finds another $B'$ that also verify correctly with the signature. 
Since the aggregate signature requires at least $n-f$ signatures, at least $f+1$ honest replicas have signed a common message.
By intersection argument, there is at most one valid aggregate signature per instance.
%
%
Because the hash function is collision resistant and there is only one aggregate signature, the underlying {\cl} for producing both $B$ and $B'$ must also be also be identical. 
Suppose $B$ and $B'$ are different at the $p$-th mini-block. 
By polynomial binding, an adversary cannot find two $b_p$ and $b_p'$ that commits to identical commitment. Hence all $f+1$ honest replicas must have cheated. 

\textbf{Correctness}: An honest replica has \event{Finished} a dispersal with a type {\event{Complete}}, then by \textbf{Availability}, the retriever must be able to retrieve some data $B'$. 
Because the dispersed data $B'$ is \property{Commitment Binding}, any retriever can only retrieve identical mini-blocks that binds to the {\cl}. For those do not bind to the {\cl}, the kzg commitment cannot match. 
Let $B$ denote the mini-blocks that a leader collects right before starting the RS encoding. If the leader is honest, then $B=B'$.


\section{Non-accountable inclusion mechanisms}
\label{sec:appendix:non-accountables}

\subsection{Proofs for {\cardQuarter}}
\label{sec:appendix:non-accountables:proofs_cardQuarter}
Now we show proofs for {\cardQuarter}. One Big difference from {\cardOne} is the usage of 1D KZG scheme shown in Figure~\ref{fig:kzg_1d} and threshold signature. All mini-blocks are first concatenated into one dimensional line, and then the leader performs RS encoding to extends $n$ chunks to $3n$ chunks. Then the leader produces only one polynomial commitment, which is used as approval message.
Another big difference comes from an increase of required security assumption of $4f+1$ in order to enable information dispersal.

\textbf{Availability}, \textbf{Correctness} follows the same proof as {\cardOne}, since there are at least $f+1$ honest replicas each holding three correct chunks. \textbf{Commitment Binding} is provided by KZG polynomial commitment scheme. Because we do not change the message sequence and timeout, \textbf{Termination} is identical. Because it is not an accountable mechanism, \textbf{Data-tampering Resistance} is not applicable. \textbf{Inclusion-\textit{n-3f}} is stated in the Section~\ref{sec:card_quarter}, and the argument is identical to proofs of {\teaspoon}.
At last, \textbf{Space-Capturing}-$2f$ are identical proofs for {\teaspoon}.

\section{{\cardLite} Protocol and Analysis}
\label{sec:appendix:card_lite}

\subsection{Doubly Homomorphic Commitment Scheme}
\label{sec:appendix:card_lite:PCS}
In this section, we provide a summary on how to use Doubly Homomorphic Commitment Scheme\cite{bunz2021proofs} in {\cardLite} to achieve $O(\log{n})$ overhead complexity while still ensuring correct RS encoding. For explicit construction of polynomial commitment scheme that allows opening with witnesses, \cite{bunz2021proofs} contains more details on the construction.

Since {\cardLite} is a 2D KZG RS construction, the leader has a {\cl} of size $n$. The censorship resistance protocol described in Section~\ref{sec:card_lite} requires each replica to have $2$ parity chunks and $2$ commitments as usual, but it additionally requires samples other columns.
The two parity chunks are used for ensuring data availability.
To ensure all sampled commitments corresponds to a single commitment polynomial of degree $n-1$, every commitment must come with a logarithmic sized proof binding to the commitment $C$.
Because otherwise, the leader can attack the system by incorrect coding, such that no retrieve can reliably retrieve data binding to the {\cl}.

Regular merkle tree vector commitment provides the first requirement (logarithmic proof size), but fails the second requirement, because there is no way to verify all $3n$ commitments are interpolating a polynomial of degree $n-1$. 
The doubly homomorphic commitment scheme allow us to perform commitment of commitments, which treats the {\cl} of 2D KZG as $n$ evaluation points, which is points from $G_1$ group in some elliptic curve. 
Then with a polynomial with $n$ evaluations, we can find out their corresponding coefficients, and get the final commitment with Pedersen Commitment scheme (see Appendix~\ref{sec:appendix:background:pedersen_commitment}) by taking the inner product between the $n$ coefficients from group $G_1$ and $n$ basis from group $G_2$. We require the underlying elliptic curves has pairing capability such that we can derive a final commitment in $G_T$, see Appendix~\ref{sec:appendix:bilinear-curve}.
%


With construction of generalized inner product argument (GIPA)\cite{bunz2021proofs}, we can use the polynomial commitment scheme with ability to open indices with witness size $O(\log{n})$.
To save verification cost, one can use a structured setup on $G_2$ similar structure in the setup of KZG~\ref{sec:appendix:background:kzg} ($g^{z^i}$. 
However computation is not our concerns. But we can use SRS to limit the size of polynomial to degree $n-1$, which eliminates the needs of low degree testing to ensure the degree polynomial is indeed low degree of $n-1$.

Otherwise, we can use a low degree testing, by repeatedly reducing the degree of polynomial by splitting even and odd monomial, then committing the combined polynomial of half of the degree for $O(\log{n})$ times, until there is a constant polynomial.

\begin{figure}[h]
\centering
\includegraphics[width=0.6\textwidth]{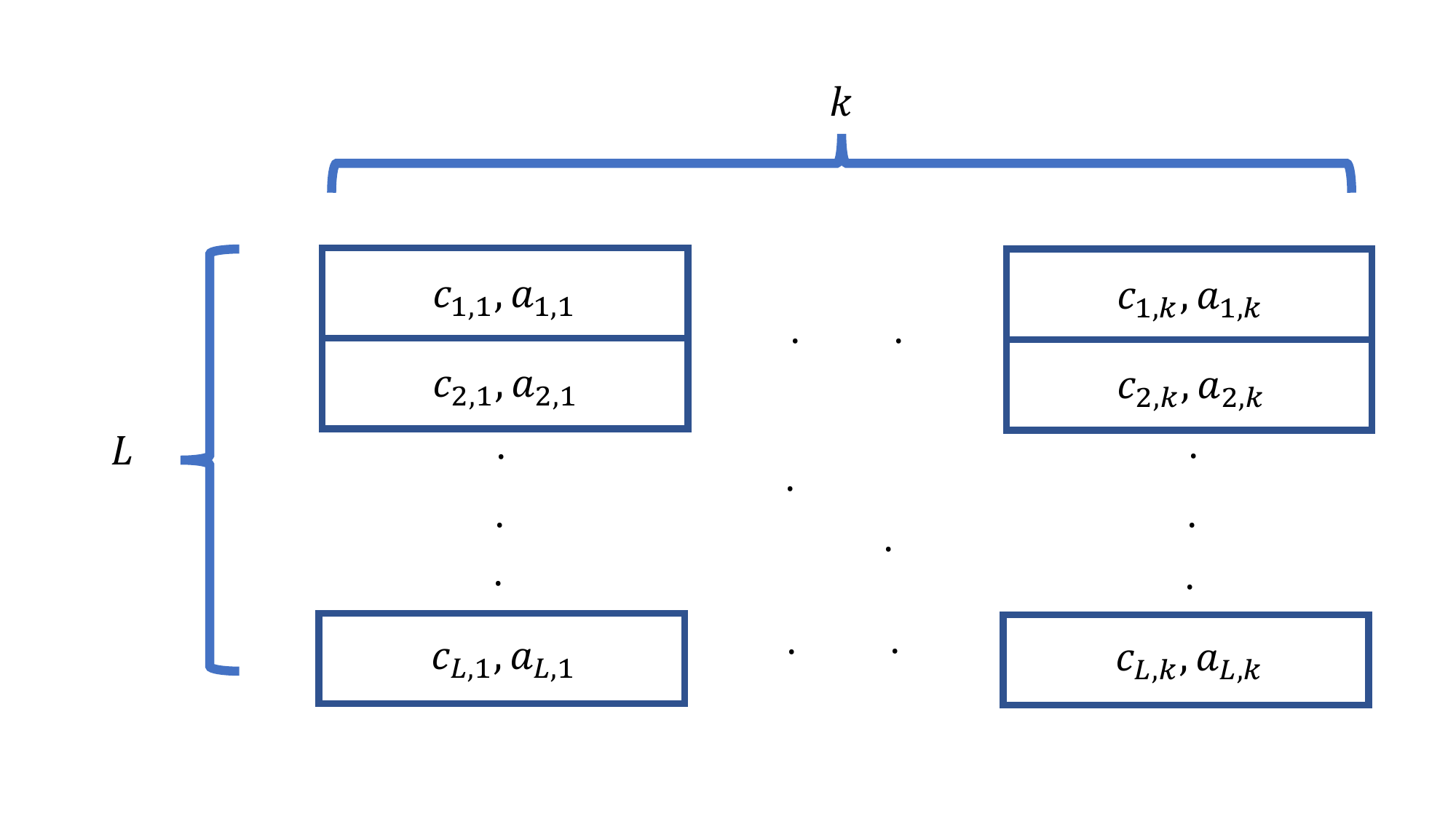}
\caption{The sampling matrix for one replica.}
\label{fig:lite_struct}
\end{figure}

\subsection{Configuring Parameter for boosting}
\label{sec:appendix:card_lite:analysis}
As described in Section~\ref{sec:card_lite}, we recap the full {\cardLite} protocol and present analysis to derive proper parameters to achieve properties stated in Table~\ref{table:card_comparison}.
Since the leader cannot send $\cl$ of length $n$ to replica, it has to use a doubly homomorphic commitment scheme to reveal commitments with $O(\log n)$ size proof.
All replicas sign the final commitment $C$ (which is commitment of the commitments) as its approving message.
After collecting mini-blocks and creating an {\al}, the leader creates $n$ sampling matrix, and each matrix has the format shown in Figure~\ref{fig:lite_struct}. 
The matrix entry at coordinate $(i,j)$ consists of a tuple of $(a_p, c_p, \textit{proof of } c_p)$ corresponding to the some random query that requests an attestation from replica $p$, where $1 \le p < n$. There are in total $kL$ random query.
In Figure~\ref{fig:lite_struct}, we use subscript to denote each sample, and omits their proofs.
Clearly, if some replica $q$ wants to create such a matrix, it cannot send $kL$ random queries to create the matrix. We require the leader to use {Fiat-Shamir} Transform to generate randomness for replica $q$ based on the commitment $C$ by using hash function $h$ to generate ${r_{i,j}} = h(C, i, j, q)$, where we use ($r_{i,j}$ mod $n$) to decide which replicas the random  query wants to receive the tuple for.
The leader has to present $n$ matrix for each replica. Let's now focus on one such matrix.

If a leader excludes attestations from some replicas, some entries in the matrix must be empty, because no leader can provide proof for something that is not on the polynomial or attestations is not available. It is orthogonal to Data tampering attack, which is handled by the 2D KZG construction.
Suppose there are two types of leaders, both types collect {\al} from replicas, and use 2D KZG scheme shown in Figure~\ref{fig:card_data_struct} for encoding, and generate a {\cl}. For simplicity of presentation, we don't allow any space capturing, hence all excluded mini-blocks must be $0$.

We want to understand the worst case scenario on how much a malicious leader can censor.
The honest leader type does not censor mini-blocks, whereas a malicious leader censors at least $s \ge1$ mini-blocks on top of $f$ allowable exclusion.
Now in the random query game, the probability that a random number mod $n$ picks on the included mini-block is $\Pr_m(\text{valid}) \le p_d = (2f+1-s)/(3f+1)$ for a malicious leader (inequality in case the malicious leader censors more); whereas for an honest leader is $\Pr_h(\text{valid}) = p_h = ({2f+1})/({3f+1})$ (inequality in case malicious replicas complies with the protocol).

Then if $k$ of such random queries are requested, the probability of providing all $k$ valid responses together is at $g_d \le p_d^k$ for malicious and $g_h \ge p_h^k$ for honest leader respectively. Clearly, $g_d < g_h$. As in the main text, we call the $k$ query together as a batch.
Now we can compute a value $\rho < 1$ such that ${p_d}^{\rho} = {p_h}$, and equivalently $\rho = \frac{\log(1/p_h)}{\log(1/p_d)} = \frac{\log(1/g_h)}{\log(1/g_d)}$. 
Similarly compute a value $k$, such that ${p_d}^{k} = \frac{1}{f^{\tau}}$, where $\tau$ is a small positive value $\tau > 1$, hence equivalently $k = \frac{\tau \log{f}}{\log{1/ p_d}}$.
We can bound $g_d$ and $g_h$ by 
\begin{align*}
    g_h \ge {p_h}^{k} &= {p_h}^{\frac{\tau \log{f}}{\log{(1/ p_d)}}} \\
                      &=  f^{- \frac{\log{(1 / p_h)}}{\log{(1 / p_d)}}} \\
                      &= f^{-\rho } \\
    g_d \le p_d^k &=  f^{-1}
\end{align*}
The second equality comes from logarithmic rule $a^{\log{b}} = b^{\log{a}}$. Now suppose the node sends $L$ batches to the leader.
Now define two probability $P_1, P_2$, which equal to the probability that none of $L$ batch is successful for honest and malicious leaders respectively. 

Suppose we choose $L=\tau f^{\rho } \log{f}$, where $\tau$ is some constant.
\begin{align*}
    P_1 \ge 1 - (1 - g_h)^L &= 1 - (1 - f^{-\rho })^{\tau f^{\rho }}\\
                            &= 1 - e^{-\tau \log{f}}\\
                            &= 1 - f^{-\tau}\\
    P_2 \le 1 - (1 - g_d)^L  &= 1 - (1 - f^{-1})^{-\tau  f^{\rho} \log{f}} \\
                &= 1 - e^{-\tau  f^{\rho -1} \log{f}}\\
\end{align*}

Clearly, $P_1$ approaches to $1$ for any $\tau>1$.
Whereas for $P_2$, we need additional procedure to derive bounds. 
Since exponential is continuous function, we can bring the limit to the exponent, 
\[
\lim_{f \rightarrow \infty} e^{-\tau f^{\rho -1} \log{f}} = e^{\lim_{f \rightarrow \infty} -\tau f^{\rho -1} \log{f}}
\]
we then need to calculate $\rho$ to derive bounds on $s$, such that for the malicious type the limit on the exponent goes to $0$, and consequently $P_2 \rightarrow 0$.

Since $\lim_{f \rightarrow \infty} (f^{-1/\log{f}})^{\tau} = e^{-\tau}$.
We want to get a $\rho$ such that $\rho-1$ equal to $\frac{-\tau}{\log{f}}$. By definition, $\rho = \frac{\log(1/p_h)}{\log(1/p_d)}$. 

We can approximate $\rho$ with Taylor series with some additive second order error $\frac{6}{(3f+1)^3}$, which is close to $0$ as $f$ increases.

\begin{align*}
    p_h &= \frac{2f+1}{3f+1} = 1 - \frac{f}{3f+1} \approx e^{-\frac{f}{3f+1}} \\
    p_d &= \frac{2f+1-s}{3f+1} = 1 - \frac{f-1+s}{3f+1} \approx e^{-\frac{f-1+s}{3f+1}}
\end{align*}
hence $\rho = \frac{f}{f-1+s}$. With the equality $\rho-1=\frac{-\tau}{\log{f}}$. Now compute $s$

\begin{align*}
    \log{f} &= \frac{\tau}{1 - \rho} = \frac{f-1+s}{s-1} \tau \\
        s   & = \frac{f \tau}{\log{f}} = \Theta(\frac{f}{\log{f}})\\
\end{align*}

As the result, 
\begin{align*}
    P_2 &\le 1 - e^{-\tau f^{\rho -1} \log{f}}= 1 - e^{-\tau e^{-\tau} \log{f}}  = 1 - f^{-\tau e^{\tau}} \\
\end{align*}

Clearly, $e^{-\tau}$ decreases much faster than $\tau$ for $\tau>1$, $P_2$ approaches $0$ exponentially fast as we increase $\tau$. Therefore
\begin{align*}
    P_1 &\ge 1 - f^{-\tau} \\
    P_2 &\le 1 - f^{-\tau e^{-\tau}}  \\
\end{align*}
Even with a small $\tau$, the honest replicas can differentiate two types with extremely high probability.
The leader prepare one $k \times L$ matrix for every one replica. Each matrix corresponds to a sampling replica $p$. If the leader finds a valid batch (row), it can demonstrate that sufficient many mini-blocks are included, by sending the particular valid row(batch) in the matrix to the replica.
For an honest leader, it needs to find the working batch for at least $2f+1$ replicas, hence the probability of is
\begin{align*}
    \Pr(\text{success | honest}) &= {P_1}^{2f+1} = (1 - f^{-\tau})^{2f+1} \\
    &\approx e^{-\frac{2f+1}{f^{\tau}}} \\
    &\approx e^{-2 f^{1-\tau}} \textit{  where } \tau>1\\
    &\approx 1
\end{align*}
However, if the leader is malicious, it needs to convince at least $f+1$ honest replica.

\begin{align*}
    \Pr(\text{success | malicious}) &= {P_2}^{f+1} = (1 - f^{-\tau e^{-\tau}})^{f+1} \\
    &\approx e^{-\frac{f+1}{f^{\tau e^{-\tau}}}} \\
    &\approx e^{-f^{(1-\tau e^{-\tau})}} \\
    &\approx e^{-f}\\
    &\approx 0
\end{align*}

Hence for any number of censoring with $s \ge \Theta(\frac{f}{\log{f}})$, there is extremely low chance that all $f+1$ would find valid batches.
So a malicious leader can at most censor $s = o(\frac{f}{\log{f}})$. In summary, the leader needs to create $n$ sample matrix for each replicas, each of which has  {$(\tau f^{\frac{f}{f-1+s}} \log{f})$} number of rows and ($\frac{\tau}{\log{1/ p_d}} \log{f}$) number of column.
Since $p_d = \frac{2f+1-s}{3f+1}$ and $1 \le s \le o(\frac{f}{\log{f}})$, we can derive a more concise inequality for the number of columns, i.e batch size as
\begin{align*}
    \centering
    k &= -\tau (\log{f}) \log{p_d} \\
                        & \le -\tau (\log{f})\log(\frac{2f - \frac{f}{\log{f}}}{3f})\\
                        &\le \tau (\log{f}) \log{2/3}\\
                        &\le \kappa \log{f}
\end{align*}

where $\kappa$ is a small positive constant.

For every dispersal, a leader has in total $O(f^2 \log^2{n})$ rows for all replicas.
A malicious leader cannot grind the final commitment $C$ to increase the probability of convincing $f+1$ replicas if it censors $s=\Theta(\frac{f}{\log{f}})$, 
%
because it needs to grind at least $Z$ times, such that $Z*e^{-f}$ equals to some constant, where $n$ can take asymptotic value.

The overall overhead message size is therefore just the $\kappa \log{f}$ tuples consisting of $(a_p, c_p, proof)$. The proof part takes $O(\log{n})$ size because the polynomial commitment scheme uses IPA. However, since Pedersen commitments are aggregatable, the $\log{f}$ proofs can collapse together in a single proof. Hence the overall overhead complexity is still $O(\log{f})$.

\subsection{{\cardLite} Proof}
\label{sec:appendix:card_lite:proof}
We now prove the {\DAL} properties for {\cardLite} property.

\textbf{Termination}: Since we don't touch the protocol mechanism, and there is still timeout to detect livelessness, Termination proof is identical to the {\cardOne}.

\textbf{Availability}: If an honest replicas received an threshold signature, there must be at least $2f+1$ signer, and $f+1$ are honest replicas. 
Since there are at least $f+1$ honest replicas that signed the message, eventually an honest retrieve can find a honest replicas to get the correct {\cl} that commits to the final commitment $C$.
After getting the {\cl}, since coding ratio is $1/3$, any honest retriever can reconstruct the whole BFT block without concerning about the degree of the polynomial and incorrect coding. 
This is because we use IPA for commitment to the polynomial and can use a fixed size SRS to ensure the commitment is consistent to a polynomial of degree $n-1$. (we can optionally use a low degree testing for the same purpose).

\textbf{Commitment Binding}: The idea is similar in proof for {\cardOne} except we replace the collision resistant hash function(CRHF) with a polynomial commitment scheme (generalized inner product argument), which is also collision resistant. 

\textbf{Correctness}: The proof is identical to {\cardOne}.

\textbf{Data-tampering Resistance}: Since IPA is commitment binding, no adversary can find another {\cl} that maps to the same final commitment. 
As the result, the proof is identical to {\cardOne}, since all signing replicas is using a common {\cl}.

\textbf{Inclusion-$(f-o(f/ \log{n}))$}: The proof is a direct consequence from Section~\ref{sec:appendix:card_lite:analysis}.

\textbf{Space-Capturing-($o(f / \log{f})$)}: Since the Data-tampering Resistance is $n - 2f - z$ where $z = o(f / \log{f})$, the malicious leader can capture $z$ mini-blocks without honest replicas noticing. The malicious leader cannot capture more than $z$ because the correct batch of size $k$ which the leader has to provide to all honest replicas must come with $k$ original attestations. 

One thing worth noticing is that the protocol designer also has a coarser control on how much to allow the leader to do space capturing, instead of a batch containing all $k$ original attestation, it could allow $k-1$ attestations but allow the leader to have $2$ attestations.

\section{Stealing or {\LSC}?}
\label{sec:appendix:card_design:lsc}
So far, we have yet discussed Leader Space Capturing ({\LSC}) in depth. This section serves for developing difference and clarifying some internal tradeoff.
In developing censorship resistance property, we use decentralization as the underlying force to achieve it.
It is therefore natural to adhere to a principle that block building is a collective work that everyone should be given a fair share of space resource. 
However, a space capturing is a phenomenal when a leader has taken data space from other replicas which can potentially have positive effects on the system utilization.
On the other hand, if without constraint a malicious leader can capture too much space leaving only one guaranteed inclusion of honest replica; in the protocol {\teaspoon}, as $n$ becomes larger, the protocol is as close to as a regular leader based BFT protocol.
{\LSC} is intractable with non-accountable inclusion mechanism because there is no mapping about whose mini-blocks get included. 
On the other hand, accountable inclusion mechanism allows protocol designers to tune the amount of {\LSC}. 
In both accountable mechanisms we discussed, the parameter $\eta$ is configured as $0$, however with little modification, all accountable protocols can increase the number, by allowing a leader to replace some $0$ in the {\cl} with leader's own attestations and mini-block.

Given possible choices of {\LSC}, we now elaborate in what ways {\LSC} can be beneficial and harmful, which would be helpful for protocol designers to decides if to allow {\LSC} and at what amount.

As mentioned, space capturing is unobservable in non-accountable mechanisms. 
However, if it is allowed in accountable mechanism, $\eta$ is the only mechanism which can control the amount of {\LSC}. Recall in in section~\ref{sec:card_design:freq_waiting}, we discuss a local parameter {\textit{read-proc-time}} which determines how much longer an honest leader waits after receiving $n-f$ mini-blocks. 
But since the {\textit{read-proc-time}} is an unobservable parameter to the outside parties, all leaders can set it to maximize {\LSC}, and the protocol has no means of enforcement.
On the positive side, when the network is bad and many replicas cannot submit mini-blocks fast enough before the dispersal period ends, 
{\LSC} allows the leader to make use of available space for higher system throughput. 
On the negative side, if we attempt a rational model for a moment, and assume replica's revenue increases with the number of included mini-block, then {\LSC} encourages any leader to take data space from other replicas. %
However, at the same time if everyone gets the chance to be a leader, all can do {\LSC} while at the same time increasing system utilization.

If {\LSC} is allowed, there is subtle difference of {\LSC} depending on the accountability of the inclusion mechanism. 
If non-accountable mechanism is used, a malicious leader can apply censorship on the basis of replicas, and the replica which repeatedly do not get mini-blocks included cannot even prove itself being censored. And vice versa for a malicious replicas falsely accuses a leader when its data is included.
In accountable mechanism, although we cannot attribute which party is at fault, we can at least know without false alarm there is something wrong in the protocol such that we can perform analysis on the leader or even increase network waiting time.
%

%
If we view {\LSC} from a broader perspective that takes clients into considerations, the next section shows that the {\LSC} parameters $\eta$ has a great impact on the amount of censorship resistance to client's experience and system efficiency. In general, the less is the $\eta$ value, the better is the client's experience and system efficiency.

\section{Navigator protocol}
\label{sec:appendix:navigator:protocol}

\subsection{Client and Replicas time synchronization}
\label{sec:appendix:navigator:tx_broadcast_sync} 
We present a time synchronization protocol in Algorithm~\ref{alg:navigation_sync}. Essentially, the algorithm requires the client to periodically pings each replicas about their status about when they would submit the next mini-block. Assuming the worst cast network delay $\Delta$ after GST, the client can assure they can make the deadline. If the network has not reached GST, the BFT protocol might not reach agreement. 
The accuracy of the status synchronization depends on the ping frequency, which is a tunable parameter for each client.

\begin{algorithm}[!t]
    {\fontsize{8pt}{8pt}\selectfont \caption{{\navigation} Synchronization}
    \label{alg:navigation_sync}
    \begin{algorithmic}[1]
        \State $S_v \gets \emptyset$, $v \gets \textit{0}$ $ {\cdots INF}$, $timer$
        \State \underline{\textbf{At the client:} }
        \Procedure{$cronJob$}{}
            \For{$i \gets 1 \cdots n$}
                \State $s \gets send \langle PING, v, \textit{timer.now} \rangle \quad \quad \text{\# timeout 2}\Delta$
                \State $S_v \gets S_v \cup  s$ 
            \EndFor
            \State $\textit{\# } \Delta \textit{ is latency in network model}$
            \State $S_v \gets \{s: s \in S_v, \textit{timer.now} < s.deadline \}$ 
        \EndProcedure
        \State \underline{\textbf{At the replica } $p$ :}
        \Event{$ \textit{receiving } \langle PING, v, tn\rangle $}
            \If{$(v > \textit{myView}) \lor (v = \textit{myView } \wedge \textit{currAt = SIGNAL} )$}
                \State $\textit{deadline} \gets tn + 2\Delta + \textit{timeToNextSignal}$
                \State $send \langle PONG, deadline\rangle$
            \Else
                \State $send \langle PONG, 0\rangle$
            \EndIf
        \EndEvent
        \Procedure{$\text{getAvailableReplicas}$}{ $x$, $tv$} 
            \State \Return $\textit{shuffle}(S_{tv})[x:]$ $\textit{\# } \textit{ return only x replicas}$
        \EndProcedure
    \end{algorithmic}}
\end{algorithm}  

\begin{figure}[h]
\centering
\includegraphics[width=0.5\textwidth]{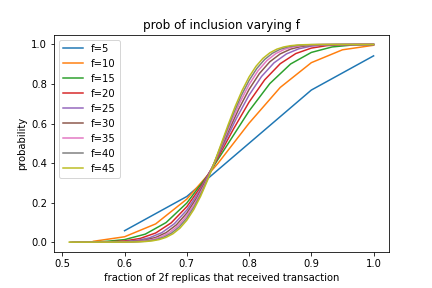}
\caption{The probability of inclusion given the leader is malicious with varying $f$. The x-axis is a ratio of the number of copies of transactions and $2f$}
\label{fig:prob_inclusion_malicious_leader_f1}
\end{figure}

\subsection{Empirical Evaluation on Inclusion Probability when Leader is honest}
\label{sec:appendix:navigator:empirical_eval_with_honest_leader}

Figure~\ref{fig:heatmap_honest_leader} plots the probability of censorship for a transaction using a {\DAL} with $t=f+1$ and $n = 3f+1$ when the leader is honest, such that $t \le q \le n-f$, where $f=10$ and $n=3f+1=31$.

\begin{figure}[h]
\centering
\includegraphics[width=1\textwidth]{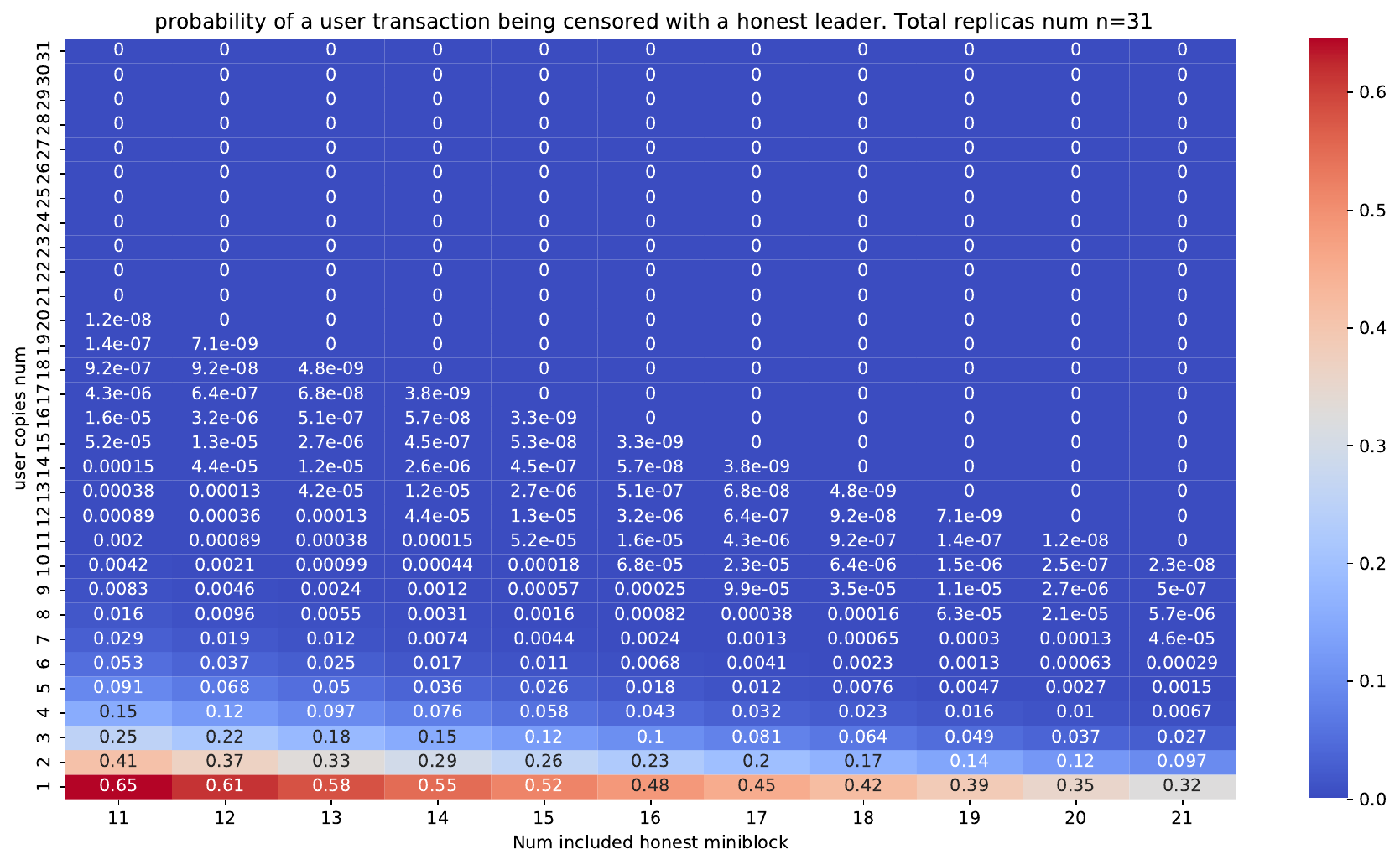}
\caption{Under honest leader, the probability of a transaction being censored as a function of $(x,q)$, where $1\le x \le n$, $f<q \le 2f+1$} 
\label{fig:heatmap_honest_leader}
\end{figure}

\subsection{Analysis on probability of inclusion given malicious leader with $f+1 \ge x$}
\label{sec:appendix:nagivator:malicious_leader_f1} 
Figure ~\ref{fig:prob_inclusion_malicious_leader_f1} plots the probability of inclusion with varying number of replicas $n=3f+1$. The inclusion can be satisfied with high probability even the client sends less than $2f+1$ copies, approximately $1.8f$. 
The probability of inclusion decreases as the clients send less number of copies of transactions.

%

\subsection{Analysis on the upper bound on the total number of censored transaction}
\label{sec:appendix:nagivator:malicious_leader_upper}
Suppose there is a total $T$ unique transactions, denoted as $t_1,..t_T$ by all clients,
where each transaction has $x$ copies which are randomly sent to replicas without replacement.
For a small constant ratio $\alpha$, it is exponentially unlikely that the leader would be able to censor more than $\alpha T$ transaction.%

Define a set of mini-blocks $S$ of cardinality $n-f$, we want to compute the probability that a malicious leader can find a particular $S$ such that the union of all transactions from mini-blocks in set $S$ has a size less than $(1-\alpha) T$.
For the same reason we discussed in Section~\ref{sec:bigdipper_prob_bad_leader}, the probability of inclusion is entirely depend on how likely a client sends its transaction to honest replicas; since adversary can always remove those transaction by selectively accepting other honest mini-block.
The probability of such set $S$ exists can be written as Equation~\ref{eq:7}, 
where every transaction corresponds to an independent random indicator variable, $\mathbbm{1}_{t_i \in S}$.
The variable is independent because every transaction runs its own instance of Algorithm~\ref{alg:navigation}.
It is a Bernoulli random variables that evaluates to 1 if condition $t_i \in S$ holds true; otherwise 0. The probability of the condition being true is $ p = 1- \prod_{i=0}^{x-1} \frac{n-f}{n-i}$. 

\begin{align*}
\label{eq:7}
    \Pr(\textit{remove } \alpha T ) &= \Pr_{ {|S| = n-f}} \Bigl\{ \exists S: \sum_{i=1}^{T} \mathbbm{1}_{t_i \in S} < (1-\alpha) T \Bigr\} \tag{7} \\
\label{eq:8}
    &\le {n \choose n-f} \Pr \Bigl\{ S: \sum_{i=1}^{T} \mathbbm{1}_{t_i \in S} < (1-\alpha) T \Bigr\} \tag{8} \\
\end{align*}
The inequality from Equation~\ref{eq:7} to~\ref{eq:8} comes from Boole's inequality.
%
%
%
To get the probability on the right high side, we use a random variable $\frac{\sum_{i=1}^{T}\mathbbm{1}_{t_i \in S}}{T} = X$, whose expected value is $X$ is $p$, by rearranging terms of Equation~\ref{eq:8}, we get 
\begin{align*}
    \label{eq:9}
    \Pr \Bigl\{ S: \sum_{i=1}^{T} &\mathbbm{1}_{t_i \in S} < (1-\alpha) T \Bigr\} = \Pr \Bigl\{ X < 1-\alpha \Bigr\} \tag{9}\\
            &=  \Pr \Bigl\{E[X] - X > E[X] + \alpha -1\Bigr\} \tag{10}\\
            \label{eq:11}
            &\le  \Pr \Bigl\{\left|E[X] - X\right| > E[X] + \alpha -1\Bigr\} \tag{11}\\
            \label{eq:12}
            &\le 2 \exp{\{-2(p+\alpha-1)^2T\}} \tag{12}
\end{align*}
In Equation~\ref{eq:11}, we let $ \alpha >1-E[X]=\prod_{i=0}^{x-1} \frac{n-f}{n-i}$, inequality is due to $\Pr \Bigl\{\left|E[X] - X\right| > E[X] + \alpha -1\Bigr\} = \Pr \Bigl\{E[X] - X > E[X] + \alpha -1\Bigr\} + \Pr \Bigl\{X-E[X] < -E[X] - \alpha +1\Bigr\}$. Inequality~\ref{eq:12} comes from Hoeffding's Inequality.

The probability of removing $\alpha T$ transactions is therefore upper bounded by 
${2 {n \choose n-f} \exp{\{-2(p+\alpha-1)^2T\}}}$.

%
%
As it can be shown that given $300$ replicas and $200$ transactions per mini-block, a malicious leader can censor more than $5\%$ of the transactions.

    
    
    

\subsection{Navigator Properties Proof}
\label{sec:appendix:navigator_proof}
In this section we prove the four properties stated in Section~\ref{sec:navigation}.

\textbf{Strict Transaction Inclusion in DA layer}: Suppose there is a {\DAL} instance with $t$ guaranteed honest mini-block inclusion, and the client sends $n-t+1$ copies of the same transaction to distinct replicas which are available for the instance round, then suppose the leader is malicious and only includes the minimal $t$ mini-blocks from honest replicas, which is required by the {\DAL}. By pigeonhole principle, there must be at least one transaction received by one of the $t$ honest replicas whose mini-block would get included. Similar argument can be made for honest leader.

\begin{figure}[h]
\centering
\includegraphics[width=0.5\textwidth]{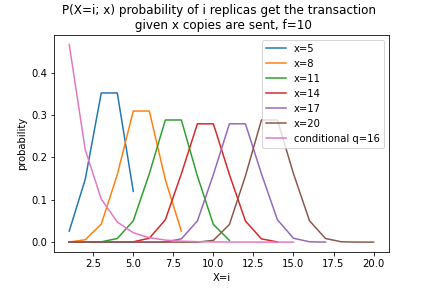}
\caption{The probability of that , $X=i$, $i$ number of honest replicas receiving the transaction given $x$ copies was sent. We show one conditional probability when $q=16$ assuming network latency is i.i.d to illustrate the point that both $q,x$ needs to work together to increase the probability of $\Pr{(IN)}$.}
\label{fig:prob_X_i_given_x}
\end{figure}

\textbf{Probabilistic Transaction Inclusion in DA layer Under Honest leader}: Suppose a {\DAL} instance with $t$ guaranteed honest mini-blocks inclusion, if an honest client sends $x$ copies of a transaction to random distinct available replicas. We first show that the probability of inclusion increases with $x$ and $t$ in the feasible region where a leader cannot exclude all $x$ copies without making the instance \event{Incomplete}. From Section~\ref{sec:navigation} and Table~\ref{table:prob_inclusion_table}, we can define the feasible region as $(0 < x < n,t\le q \le n-f)$ when the leader is honest; 
%
By our first properties, we already knows when $x\ge n-t+1$, the probability of inclusion is $1$. We now focus on the rest of the case.

When the leader is honest and $(0 < x < n-t+1,t\le q < n-f)$, Section~\ref{sec:navigation:honest_leader} shows that 
\begin{align*}
    \Pr(\textit{IN; x,q} ) 
             &= 1 - \sum_{i=1}^x \Pr(\textit{None } | X=i; q ) \Pr{(X=i;x)}\\
             &= 1 - \sum_{i=1}^x \Pr(\textit{None } | X=i; q ) \frac{{n-f \choose i}{f \choose x-i}} { {n \choose x}}\label{eq:13} \tag{13}\\
\end{align*}

where $X$ is a random variable indicating the number of honest replicas that has client's transaction. {$P(q,i) = 1- \Pr(\textit{None } | X=i; q)$} is the probability none of $X$ honest mini-block containing the client's transaction. 
We note that $P(q,i)$ is monotonically increasing in $q$, it is because as there are more guaranteed honest mini-block inclusion, the chance that one of the $i$ mini-blocks containing the transaction gets included is non-decreasing. 
Since $q \ge t$, $P(q,i)$ increases as $t$ increases. Because $\Pr{(X=i; x)}$ does not depend on $t$, the overall probability $\Pr{(IN)}$ as a function of $q$ increases as $P(q,i)$ increases.
Therefore $\Pr{(IN)}$ is monotonically increasing as $q$ increases.

We now argue  $\Pr{(IN)}$ is monotonically increasing as $x$ increases. We first show that $P(q,i)$ is monotonically increasing in $i$ while fixing $q$. 
Since the honest leader is collecting $q$ honest mini-blocks out of $n-f$ honest replicas. If the number of honest replicas whose mini-block containing the transaction increases, then the chance of $P(q,i)$ is at least non-decreasing. 
Now in the Equation~\ref{eq:13}, there is a conditional sum over all possible value which $X$ can take up to $x$; in the worst case when $\Pr(\textit{None } | X=i; q)$ is a constant ratio, then the final probability $\Pr{(IN)}$ is also be a constant. 
But as the conditional probability starts to decrease, the overall $\Pr{(IN)}$ increases.
In addition to that, as $x$ increases, the distribution of $P(X=i;x)$ gets shifted rightward, such that it is more common (higher probability) that more honest replicas receiving the transaction.
The intuition behind the phenomena is as $x$ increases, the numerator terms of $\frac{{n-f \choose i}{f \choose x-i}} { {n \choose x}}$ finds its peak also at higher $i$.
Figure~\ref{fig:prob_X_i_given_x} plots distribution for different $x$ when $f=10$.
%
Hence if the conditional probability drop fast enough, a client only needs to send a few copies $x$, such that the permutation term (binomial like) raises up just after the conditional term drops, the probability can be kept very low. 
In Figure~\ref{fig:prob_X_i_given_x}, the conditional probability is using i.i.d assumption with number of included mini-block $q=16$, out of $21$ possible honest mini-blocks with $n=31$.

Now, we show the second statement that if the network latency is i.i.d, the inclusion probability increases exponentially with $x$ and $t$. Since network latency is i.i.d,  the conditional part can be computed as 
\begin{align*}
    \centering
    \Pr(\textit{None } | X=i;q ) &= \prod_{j=0}^{q-1} \frac{n-f-i-j}{n-f-j} \\
                &= \prod_{j=0}^{q-1} \frac{a - i}{a }\\
                &< (1 - \frac{i}{n-f})^{q} \\
                &\approx e^{-q\frac{i}{n-f}}
\end{align*}

Hence we know the term {$\Pr(\textit{None } | X=i; q)$} is exponentially small when $q$ increases. However, we need also show that the summation term in Equation~\ref{eq:13} is small.

First note that, the permutation term, $\Pr{(X=i;x)}$, is always less than $1$, so we can just focus on the case when $i$ is small, because when $i$ is large, the conditional probability part is already very small.

Now we use the binomial upper bound, ${n \choose k} < (\frac{ne}{k})^k$\cite{das2016brief}, for numerator, and binomial lower bound, $(\frac{n}{k})^k < {n \choose k}$, for denominator to show that in worse case, when $i$ is small the probability is exponential decreasing when $x$ increases

\begin{align*}
    \centering
    \frac{{n-f \choose i}{f \choose x-i}} { {n \choose x}} &< \frac{(\frac{(n-f)e}{i})^i (\frac{fe}{x-i})^{x-i}}{(\frac{n}{x})^x}\\
            &= \frac{(n-f)^i f^{x-i} x^x e^x}{i^i (x-i)^{x-i} n^x}\\
            &= \frac{(n-f)^i f^x x^x e^x}{i^i x^x n^x}\\
            &\approx \frac{(n-f)^i}{i^i} (\frac{e}{3})^x
\end{align*}

Clearly for small $i$, as $x$ increases the value is upper bound by a exponentially decreasing value. Hence it completes the proof.

\textbf{Probabilistic Transaction Inclusion in DA layer Under Malicious Leader}: The first part of the statement that the probability increases as $x$ increases comes directly from the fact from Section~\ref{sec:bigdipper_prob_bad_leader} where $f+1 \le x < n-t+1$ that 
\[
    \Pr(IN) = \sum_{i=f+1}^{x} \frac{ {n-f \choose i} {f \choose  x-i} }{{n \choose x}}
\]
Clearly, $\Pr(IN)$ increases as $x$ increases. If $x\le f$, the probability is $0$, because the malicious leader can choose the rest $f+1$ honest replicas which does not have the transaction.

To show that a malicious leader cannot censor all $u$ unique transactions altogether, we use the result from Section~\ref{sec:bigdipper_upper_bound}, where $p = 1- \prod_{i=0}^{x-1} \frac{n-f}{n-i}$
\[
    \Pr(\textit{Censor all u txs}) = 2 {n \choose n-f} \exp{\{-2(p+\alpha-1)^2T\}}
\]
As we can see, given a fixed pair of $n,x$, the probability that the censoring attack succeeds is decreasing exponentially in either $T$, $\alpha$ or $x$. By letting $u= \alpha T$, we get our property.

\textbf{2} $\Delta$ \textbf{Probabilistic Confirmation in DA layer}: By Algorithm~\ref{alg:navigation}, the clients is able to get the response if some replicas are able to submit transactions. %
The client can calculate its probability of inclusion based on the response. 
Although some responses might be sent from malicious replicas, the calculation stated in the previous properties can apply directly to derive the probability result, since they already assume the malicious replicas are not including the transaction.
Since each categories have the explicit formula to compute the probability, the client can use the received the responses, the number of duplication and current belief on the leader to compute the probability of inclusion.

\section{Spamming Attack}
\label{sec:appendix_navigation_spamming}
{\name} expects clients to duplicate its transactions among replicas to increase censorship resistance. 
However, without proper mechanism clients are inclined to spam the BFT replicas for increasing the chance of inclusion. 
%
%
A direct way of controlling the spamming problem is to increase the cost of transaction. 
We could use a price mechanism, like EIP1559~\cite{roughgarden2020transaction} deployed in Ethereum, which uses posted price when the network is not congested, and otherwise use first price auction. 
But the price mechanism alone does not the problem, since the price is raised for everyone. 
The final cost for a client should also correspond to the number of copies a transaction has been sent.
%
%
%
This can be enforced by requiring the execution layer (virtual machine) to also charge fees by transactions counts as a meta information when calculating the spent resources.

Another method that does not rely on in-protocol mechanism like 1559 is to use market mechanism entirely, such that every BFT replicas is free to post its own prices. 
The $n$ replicas creates a competitive market. If we consider there are 2 types of transaction hold by various clients, and let the BFT replicas being the first mover in the information asymmetry game, a natural solution is to use screening mechanism.
We can find a separating subgame perfect Nash equilibrium if the number of high value transaction type is relatively small compared to the low value type transaction. The screening mechanism has been been well studied in the field of contract theory~\cite{rothschild1978equilibrium}. 
However in the screening model, the BFT replicas would lose some profits from high value transaction, because the high value type would mimic the type of low value transactions.
Then the replicas can use the value to sort the transaction and only include the ones that fit into the mini-blocks.

For spamming concern at network level, we can rely on standard approaches including content distribution network, proxy networks which have been already used in the real world scenario.

\begin{figure}[h]
\includegraphics[width=1\textwidth]{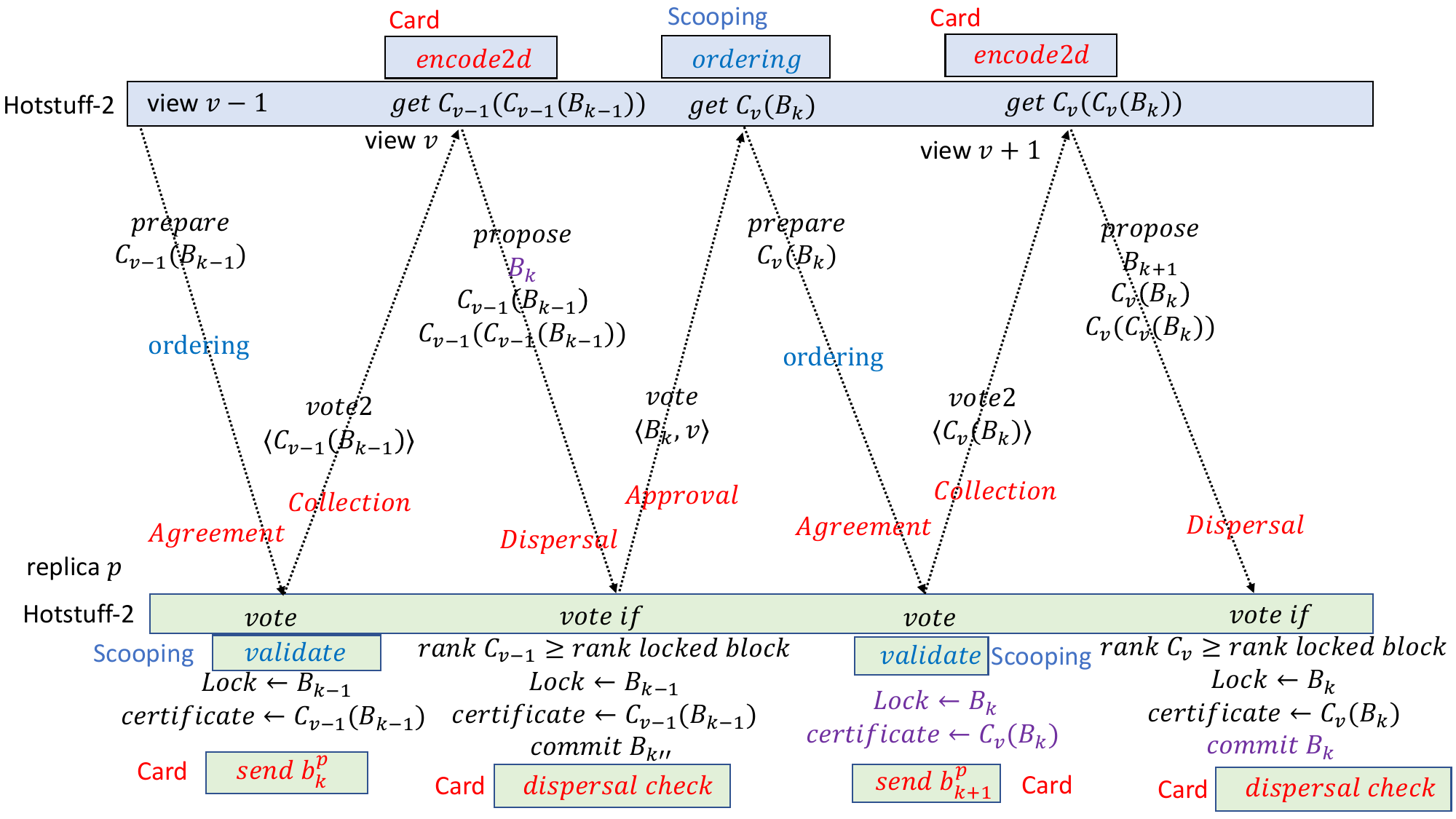}
\caption{A flow diagram depicting an adaptation of $\name$ with Hotstuff. The procedures from Hotstuff-2 protocol is marked in black; {\cardOne} protocol is marked with red, and {\scooping} protocol is marked in blue.
The signature validation procedure has to be verified after receiving the settlement message.}
\label{fig:hotstuff-scooping}
\end{figure}

\section{Integration with Hotstuff-2}

\label{sec:appendix_integration_hotstuff-2}
Figure~\ref{fig:hotstuff-scooping} depicts a flow diagram that integrates {\cardOne} with Hotstuff-2. We encourage readers to read about Hotstuff-2\cite{cryptoeprint:2023/397} protocol. $C_v(\cdot)$ represents a certificate at view $v$, $C_v(C_v{\cdot})$ represents a double certificates, which indicates the block at view $v$ can be safely committed by anyone. $k$ denotes block height, since no all view can successfully commit a block, hence $k$ is different from $v$.

On the diagram, the blue bar represents the Hotstuff-2 leader replica, the green bar at the bottom represents a non-leader replicas. Procedures needs to be carried out after receiving the messages are written close to the message arrow, and the sequence of execution is from top to down.

Because we are using verifiable information dispersal, all block $B$ from the Hotstuff-2 protocol actually means the metadata overhead described in the {\cardOne} protocol, which are {\cl} and {\al}.

All replicas need to three algorithms: Hotstuff-2 (in black and purple), {\cardOne} (in red) and Scooping (in blue).
The time flows from left to right, transitioning from view $v-1$ to $v$. Starting from \textit{Collection} message from the left, it belongs to the {\cardOne} protocol when a replica tries to get its mini-block included by the leader; 
the \textit{vote2} message belongs to Hotstuff-2 protocol, the leader needs to collect at least $2f+1$ vote2 messages before proceeding; \textit{vote2} message is sent at a stage that the replica has locked some block, and is planning to vote the second time to commit the block. The subscript indicates that the block $B_{k-1}$ at blockchain height $k-1$ happens at view number $v-1$.

When the leader collects at least $2f+1$ Collections messages and $2f+1$ vote2 messages. The leader performs the 2D RS encoding and sends the \textit{Dispersal} message of the {\cardOne} protocol, at the same time the leader also sends the \textit{Propose} message along with the double certificates to the replica. This is the official start of the new block, the block is tagged with view number $v$; in the diagram we colored all procedures relates to this block in purple.
At the same time, the previous block $B_{k-1}$ has just been able to commit after creating a double certificate by collecting signatures from $2f+1$ \textit{vote2} messages, the double certificate are piggybacked to the replicas. 

After receiving those messages at the replica side, the replica would directly commit the block that has double certificates, in addition it would indirectly commit all blocks that are prefix of the block just being committed. In the diagram, we represents them as $B_{k''}$.
The non-leader replica then needs to perform the dispersal check which corresponds to the kzg commitment checking and signature checking from {\cardOne} Protocol. After everything passes, the replica sends the \textit{Approve} message of {\cardOne} and \textit{vote} message of Hotstuff-2 to the leader. 

The leader waits to receive messages from $2f+1$ replicas and creates the combined signature for \textit{Agreement} message of {\cardOne} and similarly a certificate $C_v(B_k)$ for \textit{prepare} message of Hotstuff-2. 

At the replica side, after receiving the first certificate for $B_k$, the first step is to check if the validate procedure from the {\scooping} protocol is successful. If it is not, then the data is not available, so the replica should not lock the certificate. Otherwise, the replica proceeds with regular Hotstuff-2 procedures. 
At the same time, the leader is preparing for a new view at $v+1$, so the replica should prepare its new mini-block $b_{k+1}$ and sends to the leader. 
And we just complete a cycle of optimistic path when nothing wrong happens.

\section{Properties proofs to {\nameOne} Hotstuff-2 integration}
\label{sec:appendix_integration_proof}

\subsection{Properties proofs to Safety and Liveness}
We prove that the integrated protocol satisfies
safety and liveness properties of BFT protocol, which has great similarity to the proof in  ~\cite{cryptoeprint:2023/397}.
%
All properties listed in Section~\ref{sec:introduction} can be proved by applying properties proven in {\cardOne} and {\navigation} protocols.

\textbf{Safety}: We prove by contradiction, suppose two blocks $B_l,B_l'$ commit at the same height $l$.
Since {\DAL} has \property{Correctness}, \property{Commitment binding} and \property{Data-tampering Resistance} properties, both block must be made available and properly certified in different BFT instances, otherwise it cannot pass Alg~\ref{alg:hotstuff-scooping} line~\ref{alg:hotstuff-scooping:line:vote_and_commit}.
Suppose both are committed as result of direct commits of $B_k,B_k'$ at view $v, v'$ respectively.
Suppose $v \le v'$ (if $v=v'$, then suppose $k \le k'$), by Lemma~\ref{lemma:hotstuff2-scooping-liveness} from Appendix~\ref{sec:appendix_lemma_integration_safety}, certificate $C_{v'}(B_{k'}')$ must extend $B_k$, hence $B_l = B_l'$.
%
%
%
%
%
%
%
%

%
\textbf{Liveness}: The integrated protocol uses identical pacemaker as Hotstuff-2.
Since {\DAL} has \property{Termination} property and Hotstuff-2 is live by itself, the protocol continues to next view if a malicious leader stalls the protocol or the network is asynchronous.
When the leader is honest and the network is synchronous, in the optimistic path where the double certificate is received, the leader can collect at least $n-f$ mini-blocks (line~\ref{alg:hotstuff-scooping:line:optimistic_send_mini_block}) from replicas.  
If the honest leader is not in the optimistic path, the leader must wait for $3\Delta$ before proposing, Line ~\ref{alg:hotstuff-scooping:line:replica-enter} requires the replica to send the collection message along with the locked certificate.
Consequently the leader always have data to propose in both cases. The \property{Termination} property of {\DAL} ensures dispersal can finish with type \event{Complete}, so BFT can make progress.

\subsection{Properties proofs to Addditional Censorship resistance properties}
\textbf{Strict Transaction Inclusion}: Because {\cardOne} provides \property{Inclusion-(t)} with $t=f+1$, the strict inclusion directly come from \property{Strict transaction in DA layer} property provided by the {\navigation} algorithm. Since the integration with Hotstuff-2 is safe and live, under an honest client, a client can have guarantee of inclusion with probability $1$ by sending $n-t+1$ copies of transaction. 
If a malicious leader excludes the transaction that has been sent to $n-t+1$ replicas, also due to \property{Strict transaction in DA layer} property provided by the {\navigation} algorithm, the {\DAL} results in an \primitive{Incomplete} state, which no valid combined signature can be generated, hence the integrated {\nameOne} Hotstuff-2 integration would not proceed because no honest replica can pass the {\scooping} Algorithm, which is inserted into the consensus path in Alg.{\scooping} line~\ref{alg:hotstuff-scooping:line:optimistic_send_mini_block}.


\textbf{Probabilistic Transaction Inclusion Under Honest Leader}: The proof is almost identical as above.
Since the integration with Hotstuff-2 is safe and live, a client can send $x$ copies of transaction to probabilistically reach sufficient replicas, the rest follow through the \property{Probabilistic Transaction Inclusion in DA layer Under Honest Leader} from {\navigation}.
Since the integrated protocol is live and safe, the probability of client's transaction is included increases exponentially with $x$ and $t$.


\textbf{Probabilistic Transaction Inclusion Under Malicious Leader}: The proof is almost identical as above. If the leader is malicious and the instance \event{Completes}, we use can use \property{Probabilistic Transaction Inclusion in DA layer Under Malicious Leader} from the {\navigation} protocol, to derive that the inclusion probability is monotonically increasing with $x$, and when $x \le f$, the probability of censoring $u$ transactions altogether decreasing exponentially as $u$ increases.

$2\Delta$ \textbf{Probabilistic Confirmation for Clients}: It is a direct consequence of {\card} property and {\navigation} property from Section~\ref{sec:navigation}.

\textbf{Hyperscale Throughput}: 
Since all {\name} protocols use verifiable information dispersal, the throughput of reaching consensus to a sequence of bits in not bottleneck from downloading the whole data of size $O(bn)$, . 
The system capacity is summation of bandwidth from all replicas, which is $O(C_{band}n)$, where $C_{band}$is bandwidth for each replica. 
Then the limit on rate of information dispersal for each replicas needs to be $O(b) \le C_{band}$, so at maximum the rate of reaching consensus on a sequence of bits is only a constant of system capacity.

\subsection{Lemma in support for safety proofs of {\name} Hotstuff-2 integration}
\label{sec:appendix_lemma_integration_safety}
We show that lemma4.1 from Hotstuff-2\cite{cryptoeprint:2023/397} still holds in the integrated protocol.

\begin{lemma}
\label{lemma:hotstuff2-scooping-liveness}
If a party directly commits block $B_k$ in view $v$, then a certified block that ranks no lower than $C_v(B_k)$ must equal or extend $B_k$.
\end{lemma}
\begin{proof}
We use the same approach as the Hotstuff-2 proof, we are mostly concerned that the data is unavailable, we present the whole proof and shows the data has to be available for any committed blocks. 
Suppose we consider a block $B_{k'}'$ that is certified in view $v'$ to produce certificate $C_{v'}(B_{k'}')$ has a rank higher than $C_v(B_k)$ if $v' > v$. We will use induction on view $v'$.
At the base case $v=v'$, $C_{v'}(B_{k'}')$ is certified in view $v$. But it is impossible, since no honest replica would vote twice per a single view. 
At the induction step, since $C_v(C_v(B_k))$ exists, there must be a set of replicas $S$ whose size is no less than $2f+1$ that possess $C_v(B_k)$. 
Since a double certificate can only be created with at least $f+1$ honest replicas, we know the data $B_k$ is available.
All replicas in $S$ must have locked on $B_k$ or a block that extends $B_k$ at the end of $v$. By induction assumption, any certified block that ranks equally or higher than $v$ must equal or extend $B_k$. By the end of $v'$, the replicas in $S$ are still locked at $B_k$ or blocks that extends $B_k$. 
The leader has two ways to enter a new view. In the first case when the leader has a double certificate, then at least $f+1$ honest replicas have ensured the data is available, so it is safe to extend on top of it.
In the second case, the leader finds out the highest certified block by waiting $3\Delta$ to receive all honest replicas. 
Since the leader verifies that all locked certificate must have a valid aggregate signature corresponding to it, this step is ensured by Alg~\ref{alg:hotstuff-scooping} line ~\ref{alg:hotstuff-scooping:line:leader_propose}.
So we know it is safe for the leader to build on top of it.
Resume the Hotstuff-2 proof. Consider a proposal at view $v'+1$, if the leader makes a proposal $B_{k'}'$ that does not extend $B_k$, its certificate must be ranked lower than $C_v(B_k)$. As the result, no honest replica would vote for it, so the certificate $C_{v'+1}(B_{k'}')$ cannot be formed. Hence if $B_{k'}'$ was committed, $C_{v'+1}(C_{v'+1}(B_{k'}'))$ is formed, $B_{k'}'$ must extend $B_k$.
\end{proof}

To clarify special cases, when a malicious leader correctly runs the {\card} protocol but refuses to broadcast the aggregate signature to anyone, the malicious leader eventually gets replaced. However, since no one has the aggregate signature, a new block proposed by an honest leader would not build on top of that block after waiting $3\Delta$. 
If a malicious leader broadcasts the aggregate signature, but refuses to create the double certificates, any honest replica can send the locked certificate along with the aggregate signature to convince the leader to build on top of it.

\subsection{Signal Frequency and Collection Waiting Time}
\label{sec:card_design:freq_waiting}

Each {\DA} instance starts when replicas receive a local \event{Signal}. In the
HotStuff-2 integration, this signal is supplied by the BFT pacemaker and carries
the local view number. We run one {\DA} instance per BFT view. This frequency is
important. If multiple independent {\DA} instances were allowed for the same BFT
view, a malicious leader could selectively complete one instance while ignoring
others, weakening the intended inclusion rule without stalling the BFT protocol.
If {\DA} instances were less frequent than BFT views, then the chain could not
consume newly dispersed data at the pace of block production.

The leader also needs a rule for when to stop waiting for \textsf{Collection}
messages and begin dispersal. In Algorithm~\ref{alg:card}, this rule is denoted
\textit{reach-proc-time}. It can be implemented as a short local timer that
becomes true after the leader has waited long enough, in addition to the
requirement that the leader has received at least \(n-f\) mini-blocks. If
\textit{reach-proc-time} is always true, the leader proceeds as soon as it has
\(n-f\) mini-blocks. If the timer is longer, the leader may collect more honest
mini-blocks before dispersal.

This choice affects the latency/inclusion tradeoff. More collected honest
mini-blocks improve the transaction-broadcast guarantees from
Section~\ref{sec:navigation}: a client can achieve the same inclusion
probability with smaller fanout when more honest mini-blocks are captured.
However, waiting longer also weakens optimistic responsiveness. The timer is
therefore a performance parameter, not a safety condition. Regardless of the
timer setting, replicas only approve and later scoop a {\DAL} instance if the
certificate and inclusion checks are valid.

\end{document}